\documentclass[letterpaper]{article} 
\usepackage{aaai23}
\usepackage{times} 
\usepackage{helvet} 
\usepackage{courier} 
\usepackage[hyphens]{url} 
\usepackage{graphicx} 
\usepackage{natbib} 
\usepackage{caption} 
\usepackage{mathtools}
\usepackage{braket}
\usepackage{cancel}

\usepackage{algorithm}
\usepackage{algpseudocode}
\usepackage{xcolor}
\usepackage{dsfont}
\usepackage{hyperref}
\usepackage{soul}

\newcommand{\tr}{\textrm{tr}}

\def\z{\bold{0}}

\usepackage{amsmath,amssymb,amsthm}

\newtheorem{theorem}{Theorem}
\newtheorem{lemma}{Lemma}
\usepackage{caption}
\usepackage{subcaption}

\usepackage{newfloat}
\usepackage{listings}
\makeatletter
\newcommand{\setlabel}[1]{\edef\@currentlabel{#1}\label}
\makeatother
\DeclareCaptionStyle{ruled}{labelfont=normalfont,labelsep=colon,strut=off} 
\floatstyle{ruled}
\newfloat{listing}{tb}{lst}{}
\floatname{listing}{Listing}

\nocopyright 

\title{Alternating Layered Variational Quantum Circuits Can Be Classically Optimized Efficiently Using Classical Shadows 

}

\title{Alternating Layered Variational Quantum Circuits Can Be Classically Optimized Efficiently Using Classical Shadows}
\author {
    Afrad Basheer\footnote{Email: Afrad.M.Basheer@student.uts.edu.au},\textsuperscript{\rm 1}
    Yuan Feng, \textsuperscript{\rm 1}
    Christopher Ferrie, \textsuperscript{\rm 1}
    Sanjiang Li \textsuperscript{\rm 1}
}
\affiliations {
    \textsuperscript{\rm 1} Centre for Quantum Software and Information, University of Technology Sydney, NSW 2007, Australia\\
}

\usepackage{bibentry}

\begin{document}

\maketitle

\begin{abstract}
Variational quantum algorithms (VQAs) are the quantum analog of classical neural networks (NNs). A VQA consists of a parameterized quantum circuit (PQC) which is composed of multiple layers of ansatzes (simpler PQCs, which are an analogy of NN layers) that differ only in selections of parameters. Previous work has identified the alternating layered ansatz as potentially a new standard ansatz in near-term quantum computing. Indeed, shallow alternating layered VQAs are easy to implement and have been shown to be both trainable and expressive. In this work, we introduce a training algorithm with an exponential reduction in training cost of such VQAs. Moreover, our algorithm uses classical shadows of quantum input data, and can hence be run on a classical computer with rigorous performance guarantees. We demonstrate 2--3 orders of magnitude improvement in the training cost using our algorithm for the example problems of finding state preparation circuits and the quantum autoencoder.
\end{abstract}

\section{Introduction}
The past several years have seen tremendous progress in quantum hardware and one of the state-of-the-art quantum computers---IBM Eagle---now has 127 qubits~\cite{Chow2021}. However, it is widely believed that general-purpose fault-tolerant quantum computers are unlikely to happen soon. With a limited number of (noisy) qubits and non-perfect gates, near-term quantum devices like IBM Eagle can only run quantum circuits with small depth. One may wonder if quantum advantage can still be demonstrated in practical applications using these near-term quantum devices. Variational Quantum Algorithms (VQAs) are the most promising proposal for this purpose. VQAs encode a task in a parametrized quantum circuit (PQC), evaluate it using a near-term quantum computer, and optimize the parameters with a classical optimizer. For almost all applications for which quantum computers are desired, VQAs have been proposed. We refer the readers to \cite{Cerezo2021_VQA} for an excellent survey of VQAs.

VQAs can be regarded as the quantum analog of classical neural networks (NNs). A VQA is composed of multiple layers of ansatzes, which are smaller PQCs with a fixed architecture. An ansatz is analogous to a classical NN layer, and a VQA uses ansatzes that have the same architecture and differ only in selections of parameters. Many different ansatzes have been proposed in the literature. Popular choices include the Hardware Efficient Ansatz~\cite{Kandala2017}, Quantum Alternating Operator Ansatz (QAOA)~\cite{Hadfield2019}, etc.

Similar to classical NNs, recent works show that VQAs also have the problem of trainability. Indeed, vanishing gradients, also called \textit{barren plateaus}, were first theoretically demonstrated for deep PQCs by \cite{Mcclean2018} and then for shallow PQCs by \cite{Cerezo2021}. We may need to impose restrictions on the ansatz and/or the cost function to avoid this barren plateaus issue. One such solution is proposed in \cite{Cerezo2021} using the alternating layered ansatz, which has a brick-like structure and its two-qubit gates act on alternating pairs of neighboring qubits (see Fig.~\ref{fig:ala} for a simple example). It is proved that barren plateaus can be avoided for alternating layered VQAs provided that the depth of the PQC is $O(\log n)$, where $n$ is the number of qubits, and the cost function is defined with local observables. Surprisingly, \cite{Nakaji2021} recently proved that the shallow alternating layered ansatz is almost as expressive as the Hardware Efficient Ansatz. Thus the alternating layered ansatz is both expressive and trainable. In addition, this ansatz has been investigated or implemented in works such as~\cite{Hinsche2021,Wu2021,Arrasmith2021,Slattery2022}. In particular,~\citep{Cerezo2021} has shown that quantum autoencoders can be optimized without any barren plateaus for up to 100 qubits using the alternating layered ansatz. 

This work introduces a training algorithm with an exponential improvement in the number of copies of input states consumed during training an alternating layered VQA with shallow depth and local observables, where, and in the remainder of this paper, the number of copies of the input state equals the number of executions of the quantum device. The reason is, during each execution, we have to measure the quantum systems which inevitably destroy the quantum states. Thus a fresh copy of the input state is needed for each iteration.
Moreover, the training can be done entirely on a classical computer efficiently (with computational cost depending only polynomially on $n$) without the need to implement the PQC on a quantum device. This result is achieved by using the recently proposed classical shadow technique \cite{Aaronson2018,Huang2021} for quantum state tomography, and working in the Heisenberg picture rather than the Schr\"{o}ndinger model.

Specifically, for an alternating layered ansatz $U(\boldsymbol{\theta})$, an input state $\rho$ and an observable $O$, the VQAs of our interest estimate each evaluation of functions of the form 
\begin{align} \label{eq:vqa-function}
    f_{\rho, O}(\boldsymbol{\theta}) = \text{tr} (O U(\boldsymbol{\theta}) \rho U(\boldsymbol{\theta}) ^ {\dag} )
\end{align} 
using quantum computers. In contrast, our method can efficiently compute this classically on classical shadows of $\rho$. But note that all VQAs that uses alternating layered ansatzes need not have this specific form.

Our method, called \textit{Alternating Layered Shadow Optimization}, or simply ALSO, outperforms standard alternating layered VQA in two aspects:
\begin{enumerate}
    \item \emph{Exponential savings on input state copies.} Note that the number of copies of the input state needed in the standard VQA scales linearly in the total number of function evaluations required. In contrast, to achieve a similar precision, ALSO only uses logarithmically many copies. This allows for more iterations and better approximations in the classical optimization algorithm for a given PQC. In addition, it allows for more hyperparameter tuning with very few copies of the input state, and the same set of shadows can be used for multiple similar optimization problems and alternating layered ansatzes.

    \item \emph{Easy implementation on quantum hardware.} ALSO only requires the quantum device to be able to carry out single-qubit Pauli basis measurements on the input states. But standard VQA requires the ability to apply CNOT gates and rotation gates on them, and measurement also.

\end{enumerate}

    \begin{figure} 
    	\begin{center} 
    		\includegraphics[width = \columnwidth]{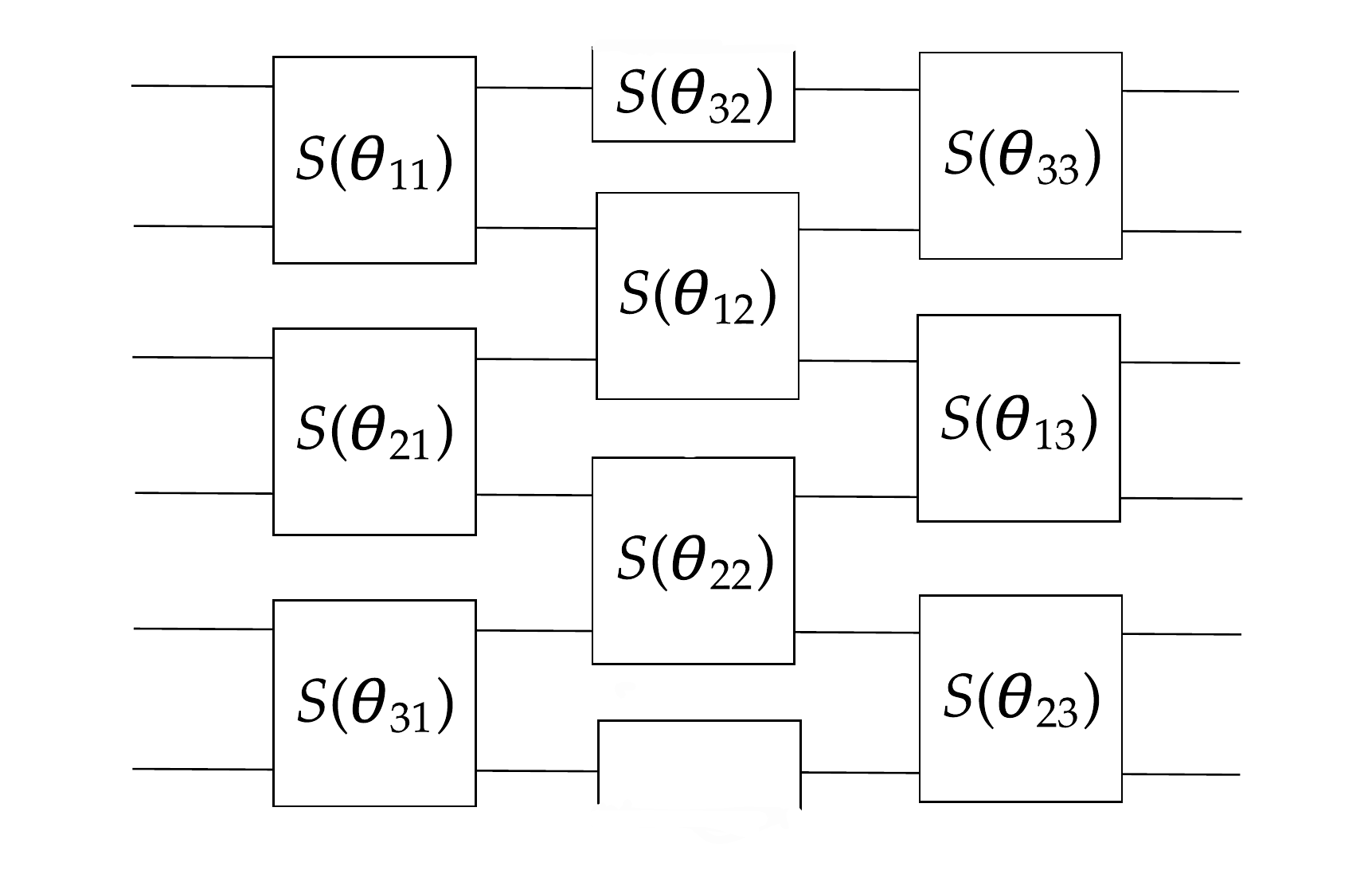}
    		\caption{An illustration of alternating layered ansatzes where the parameterized sub-circuit $S(\boldsymbol{\theta}_{32})$ is applied on the first and the last qubits. Here, $\boldsymbol{\theta}$ is an order $3$ tensor with each $ \boldsymbol{\theta}_{ij}$ being vectors of real parameters.
    		}
    		\label{fig:ala}
    	\end{center}
    \end{figure}

We demonstrate the practical efficacy of our result with two important examples: finding state preparation circuits and quantum data compression using a quantum autoencoder. In both cases, we demonstrate that ALSO can match the results of the impossible ideal VQA that uses infinite copies, using a comparatively small number of copies of the input quantum state. We also show that, with the same number of copies of the input state, ALSO outperforms the standard VQA  significantly.

\section{Related Works} \setlabel{Related Works}{sec:related_works}

    The idea of using classical algorithms, specifically classical machine learning algorithms, on classical shadows of quantum states, has been considered in~\cite{Huang2021_power,Huang2021}, where popular classical machine learning algorithms such as Support Vector Machines and Convolutional Neural Networks are trained on classically loaded shadows to solve certain important problems in quantum many-body physics. In these works, classical shadows act as low dimensional embeddings of quantum states extracted and generated using a quantum computer, and the aim of the classical post-processing is to learn powerful classifiers and prediction models. In comparison, the post-processing of classical shadows in ALSO aims to implement a VQA in a more resource efficient manner. 
    
    Variational shadow quantum circuits, developed in  \cite{Li2021},  extract local classical features by focusing on a series of local subcircuits (generated in a convolutional way and called \textit{shadow circuits}). Inspired by the classical shadow work, the approach itself does not use classical shadows. We remark that this approach can also be implemented on classical shadows of the input quantum states thanks to the locality of those shadow circuits. We expect that results similar to the ones presented in \cite{Li2021} could be achieved using an exponentially small number of copies of the input state. 

    Recent years have also seen the rise of dequantization techniques~\cite{Tang2019,Tang2021,Gilyen2018,Gilyen2022}, that is, classical algorithms that challenge the claim of quantum algorithms providing exponential speedup in solving certain classical problems. Even though our work can simulate the training of an alternating layered VQA very efficiently on a classical computer, we are dealing with a problem with quantum inputs, and still need to generate classical shadows of the input states. So ALSO cannot be considered as a dequantization of alternating layered VQAs.

\section{Background} \setlabel{Background}{sec:background}
This section recalls background in quantum computing, varational quantum algorithms and classical shadows. 

\subsection{Preliminaries}
    We use `ket' notation such as $\ket{\psi}$ to represent complex column vectors. For any $\ket{\psi} \in \mathbb{C}^d$, $\bra{\psi}$, called `bra', is its complex conjugated transpose. For any $i \in \{ 0,1,\dots,d-1\}$, $\ket{i} \in \mathbb{C}^d$ is the $i^{\text{th}}$ standard basis vector. Denote by $\mathcal{L}(\mathbb{C}^d)$ the set of all linear operators acting on $\mathbb{C}^d$.

    \subsection{Quantum computing}
    A \textit{qubit} is the fundamentally implementable entity in quantum computing. A \textit{state} is defined as any positive semi-definite operator $\rho \in \mathcal{L}(\mathbb{C}^d)$ such that $\text{tr}(\rho) = 1$ (all the diagonal elements sum up to $1$). A qubit can admit any quantum state $\rho \in \mathcal{L}(\mathbb{C}^2)$ as its value, similar to how a \textit{bit} in classical computing can admit any value in $\{ 0,1\}$. A quantum \textit{system} or \textit{register} is an array of qubits. To describe the state of a system of two qubits, $[q_1, q_2]$, we use states that act on the \textit{tensor product} of the two $2$-dimensional vector spaces, denoted as $ \mathbb{C}^2 \otimes \mathbb{C}^2 \cong \mathbb{C}^4$. So, a system of $n$-qubits can be in any state in $\mathcal{L}(\mathbb{C} ^ {2^n})$.  
    A state $\rho$ is called a \textit{pure state} if its rank is $1$. In this case, $\rho$, as well as its dynamics, can be fully characterized by any normalized eigenvector associated with the eigenvalue 1. 
    
    An $n$-qubit \textit{quantum gate} is defined as a unitary operator $U \in \mathcal{L}(\mathbb{C} ^ {2 ^ n})$. The application of such a quantum gate on a system in a state $\rho \in \mathcal{L}(\mathbb{C} ^ {2 ^ n})$ transforms the state of the system from $\rho$ to another state $U\rho U ^{\dag}$.     
    Some important gates that feature extensively in this work are 
    
    \begin{align*}
            H = \frac{1}{\sqrt{2}}
            \begin{bmatrix}
                1 & 1 \\
                1 & -1
            \end{bmatrix}, \  
             S = 
            \begin{bmatrix}
                1 & 0 \\
                0 & i
            \end{bmatrix},\ 
             \text{CNOT} = \bigg[
            \begin{smallmatrix}
                1 & 0 & 0 & 0 \\
                0 & 1 & 0 & 0 \\
                0 & 0 & 0 & 1 \\
                0 & 0 & 1 & 0
            \end{smallmatrix}
            \bigg].
    \end{align*}
    Let $Z = \ket{0}\bra{0}-\ket{1}\bra{1}$,  $X=HZH$, and $Y=iXY$. These single-qubit gates are called \textit{Pauli operators}. For $P \in \{ X,Y,Z\}$ and $\theta\in \mathbb{R}$, the $P$-rotation gate $R_P(\theta)$ is the gate $\cos(\theta/2) \mathds1 + i \sin(\theta/2) P $, where $\mathds1$ is the identity operator.
    
    An $n$-qubit \textit{observable} is defined as any Hermitian operator $O \in \mathcal{L}(\mathbb{C}^{2^n})$. To ``observe'' information from a system prepared in a state $\rho$, we \textit{measure} $\rho$ using $O$. The measurement result is probabilistic and its expected value is $\text{tr}(O\rho)$. 

An $n$-qubit operator $V \in \mathcal{L}(\mathbb{C}^{2^n})$ on register $I=[q_1,\ldots,q_n]$ 
is \textit{$k$-local} if there is a $k$-qubit sub-register $A=[q_{i_1}, q_{i_2},\dots,q_{i_k}]$ such that $V = \widetilde{V} \otimes \mathds1$, where $ \widetilde{V} \in \mathcal{L}(\mathbb{C}^{2^k})$ acts on qubits in $A$ and $\mathds1$ is the identity operator on $I\setminus A$. For simplicity, we often write $V$ as 
$\widetilde{V}[A]$ and say that $V$ acts \textit{non-trivially} on qubits only in $A$. 

    \subsection{Classical shadows using Pauli basis measurements} \label{subsec:classical_shadows}
    
   Let $\rho$ be an $n$-qubit quantum state and let $O_1, O_2, \dots, O_M$ be arbitrary $n$-qubit observables the classical descriptions of which are given. 
   Using conventional quantum tomography techniques, $O(2^n)$ copies of $\rho$ are required to estimate 
   $\text{tr}( O_i\rho)$ for each $O_i$.
   
        The classical shadow technique \cite{Huang2020},  developed from  shadow tomography \cite{Aaronson2018},  provides succinct classical descriptions of quantum states. Using this technique, $\tr(O_i \rho)$ can be collectively predicted by consuming only $\mathcal{O}(\log M)$ copies of $\rho$. Moreover, when these observables belong to certain classes, the dependency on $n$ is polynomial or constant.

        When the observables are all local (with locality $k \ll n$), the classical shadow method reduces to a very simple protocol. The first step is to measure the individual qubits of $\rho$ on a random Pauli basis. To this end, for each qubit $i$, we apply a gate $U_i$ uniformly randomly chosen from $ \left\{ \mathds1, H, HS ^ {\dag} \right\}$, and then measure it in the computational basis. Let the measurement outcome be $u_i\in\set{0,1}$. Then a
        \textit{single-qubit classical shadow} of $\rho$ is calculated (classically) as 
        \begin{equation}\label{eq:shadow}
        \hat{\rho} = F(U_1 ^ {\dag} \ket{u_1} \bra{u_1} U_1) 
                 \otimes \dots
                \otimes F \left( U_n ^ {\dag} \ket{u_n} \bra{u_n} U_n \right),
        \end{equation}
        where $F(V) = 3V - \mathds1$.  As a fully separable quantum state, the state $\hat{\rho}$ can be stored efficiently as $n \ 2\times 2$ matrices. Furthermore, $\hat{\rho}$ gives an unbiased estimation of the unknown state $\rho$ and hence $\text{tr}(O_i \hat{\rho})$ is an unbiased estimator of $ \text{tr}( O_i\rho)$ for all $i$. 
        
        Specifically, we have:

        \begin{theorem}~\cite{Sack22} \label{th:shadows}
            Let $\rho \in \mathbb{C} ^ {2 ^ n}$ be a quantum state. Suppose $O_1, O_2, \dots, O_M \in \mathbb{C} ^ {2 ^ n}$ are $M$ $k$-local observables. For any $\delta, \epsilon \in \left( 0, 1\right)$, let $T$ be any integer not smaller than $\frac{4 ^ {k + 1}}{\epsilon ^ 2} \cdot  \log(\frac{2M}{\delta}) \max_{i} \| O_i \|_{\infty}^2$
            and define shadow state $\hat{\rho}_T$ as $ \hat{\rho}_T = \frac{1}{T} \sum_{j = 1}^T \hat{\rho} ^ {(j)}$, where  $ \hat{\rho} ^ {(j)}$ are single-qubit classical shadows as in Eq.~\eqref{eq:shadow}. Then, with probability at least $1-\delta$ and for all $i$, we have $\left| \tr ( O_i\hat{\rho}) - \tr ( O_i\rho) \right| \leq \epsilon. $
        \end{theorem} 
    Note that the original version of this theorem required $\| O_i\|_{\infty} \leq 1$ for all $i$. However, this can be relaxed by dividing every matrix by $\max_i \| O_i\|_{\infty}$ and then estimating with precision $\epsilon/\max_i \| O_i\|_{\infty}$.

    Moreover, each estimation $\text{tr}(O_l\hat{\rho}_T)$ can be classically computed very efficiently. Let $A_l = [q_{l_1}, \dots, q_{l_k}]$ be the sub-register that $O_l$ acts non-trivially on and $O_l = \widetilde{O}_l \otimes \mathds1$ with $ \widetilde{O}_l \in \mathcal{L}(\mathbb{C}^{2^k})$. For the shadow $\hat{\rho}$ in Eq.~\eqref{eq:shadow}, we only need to use the $k$ $2\times 2$ matrices corresponding to the sub-register $A_l$. Denote by $\hat{\rho}_T \big|_{A_l}$ the classical shadow obtained by taking average of $T$ such reduced shadows. Then we have $ \text{tr}(O_l\hat{\rho}_T ) = \text{tr}(\widetilde{O}_l \hat{\rho}_T \big|_{A_l} )$ and hence it can be computed with cost exponential only in $k$ and independent of $n$.

    \subsection{Variational quantum algorithms} \label{subsec:vqa}

A VQA encodes the task under consideration as a parameterized quantum circuit (PQC), which is typically composed of multiple layers of ansatzes, i.e., smaller PQCs with the same architecture. Write $U(\boldsymbol{\theta})$ for the PQC, where $\boldsymbol{\theta}$ is a real-valued vector of parameters. The VQA uses $U(\boldsymbol{\theta})$ to estimate a target function's value and gradient at some point in its domain, and then optimizes the parameters of the PQC by feeding the circuit's output to a classical optimizer.

In this paper, we focus on the basic function as specified in Eq.~\eqref{eq:vqa-function} and aim to find the parameters which maximize it. Note that $f_{\rho, O}(\boldsymbol{\theta})$ can be estimated through multiple measurements. Using techniques such as the parameter shift rule~\cite{Mitarai2018}, finite differences, etc, one can also estimate the gradient of $f_{\rho, O}$ using a quantum computer. With this, we can find the optimum values of $\boldsymbol{\theta}$ by using any classical optimization method. Many tasks such as variational quantum eigensolver, finding state preparation circuits, quantum autoencoder, etc can be reduced to finding the best $\boldsymbol{\theta}$ which maximizes some $f_{\rho, O}(\boldsymbol{\theta})$.

\section{Alternating Layered Shadow Optimization} \setlabel{Alternating Layered Shadow Optimization}{sec:also} 

    In this section, we explain the key idea and theoretical results behind ALSO. Following \cite{Cerezo2021}, we request the observables to be $k$-local or linear combinations of a small number (polynomially dependent on $n$) of $k$-local observables for some $k\ll n$.

    \subsection{Alternating Layered Ansatz} 
    \setlabel{Alternating Layered Ansatz}{subsec:ala} 
    
        The Alternating Layered Ansatz is the brick-like circuit structure presented in Fig.~\ref{fig:ala}, where each  $S(\boldsymbol{\theta}_{ij})$ is a parameterized circuit acting on a small number of qubits. A simple example of $S$ is given in Fig.~\ref{subfig:s}. This work assumes that each  $S$ acts on two qubits and has $p$ real parameters, but our idea can be easily extended to the general case.
        In Fig.~\ref{fig:ala}, the total number of vertical blocks of $S$ gates, written  $d$, is called the \textit{depth} of the ansatz. The circuit depicted in the figure has $d=3$. 
        
        Let $ \oplus$ denote addition modulo $n$. Then, in a specific vertical block $j$, each circuit $S(\boldsymbol{\theta}_{ij})$ acts on qubits $2(i-1) \oplus j$ and its neighbor $2(i-1) \oplus j \oplus 1$. So, the final form of the circuit is given as 
        \begin{align}\label{eq:ala}
           \hspace*{-2mm} U(\boldsymbol{\theta}) = \prod \limits_{j = 1} ^ {d} \prod \limits_{i = 1} ^ {n/2}  S(\boldsymbol{\theta}_{ij})[2(i-1) \oplus j, 2(i-1) \oplus j \oplus 1],
        \end{align}
        where $\boldsymbol{\theta} \in \mathbb{R} ^ {\frac{n}{2} \times d \times p}$ is a tensor of real parameters where $\boldsymbol{\theta}_{ij}$ is a $p$-dimensional real vector of parameters and $S(\boldsymbol{\theta}_{ij})[k,l]$  means $S(\boldsymbol{\theta}_{ij})$ acting on qubits $k$ and $l$. 
    
    \begin{figure} 
    	\begin{center} 
    		\includegraphics[width=\columnwidth, page=1]{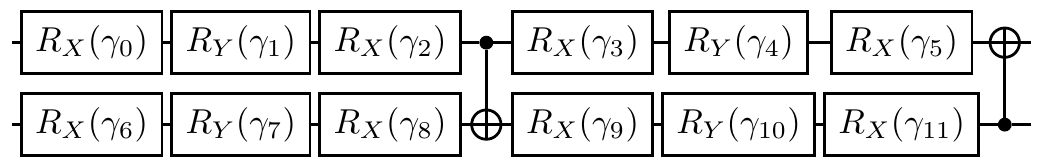}
    		\caption{The structure of $S(\boldsymbol{\gamma})$ used in the simulation. The two-qubit gate used here is the CNOT gate.
    		}
    		\label{subfig:s}
    	\end{center}
    \end{figure}
    
    \subsection{Method}
    We first explain our approach in a simpler model, with $1$-local observables and alternating layered ansatzes built with $2$-local circuits, and then extend the results to circuits and observables with arbitrary localities.

    We start with a lemma that forms the backbone of ALSO. 

\begin{figure}[t]
    \centering
    \includegraphics[width=\columnwidth]{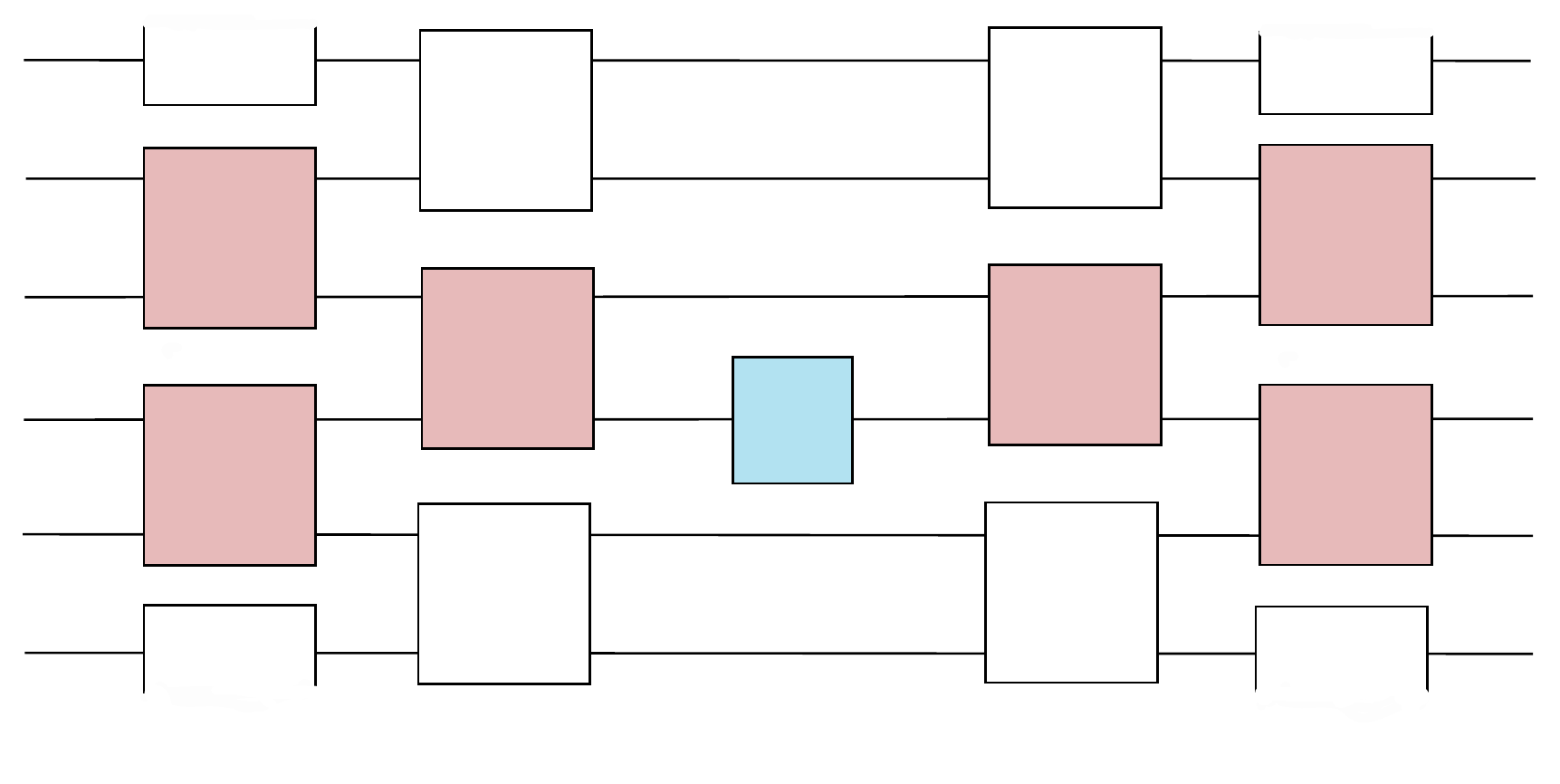}
    \caption{The structure of $W_O(\boldsymbol{\theta}) = U(\boldsymbol{\theta}) ^ {\dag} O U (\boldsymbol{\theta})$ where the blue box is a $1$-local observable, and all other boxes are $S$ sub-circuits. Except for the red ones, all other sub-circuits cancel each other out, resulting in $ W_O(\boldsymbol{\theta})$ being $2d$-local.}
    \label{fig:w_theta}
\end{figure}

    \begin{lemma} \label{le:also}
Let $d, S $ and $U$ be defined as in Eq.~\eqref{eq:ala}, and $\boldsymbol{\theta} \in \mathbb{R} ^ {\frac{n}{2} \times d \times p}$. For any $n$-qubit $1$-local observable $O$, define $W_O(\boldsymbol{\theta}) = U(\boldsymbol{\theta}) ^ {\dag} O U(\boldsymbol{\theta})$. Then we have  $\|W_O(\boldsymbol{\theta})\|_{\infty} = \|O \| _{\infty}$ and $W_O(\boldsymbol{\theta})$ is $2d$-local, that is,
$W_O(\boldsymbol{\theta}) = \widetilde{W}_{O}(\boldsymbol{\theta})[A]$ for some sub-register $A$ of $2d$ qubits. 

    \end{lemma}
    \begin{proof}
        For any $\boldsymbol{\theta}$, $ W_O(\boldsymbol{\theta})$ is obtained by conjugating $O$ with a unitary matrix. This means that the eigenvalues of $O$ and $W_O(\boldsymbol{\theta})$ are the same. So, $\|W_O(\boldsymbol{\theta})\|_{\infty} = \| O\|_{\infty}$. Figure~\ref{fig:w_theta} shows the structure of $W_O(\boldsymbol{\theta})$. 
If $d = 1$, then $W_O(\boldsymbol{\theta})$ will be an observable with locality $2$ as all blocks of the parameterized circuit except the ones acting on the qubit where the observable acts on will cancel out. Similarly, if $d = 2$, then $W_O(\boldsymbol{\theta})$ will have locality $4$. For each increment in $d$, the locality of $W_O(\boldsymbol{\theta})$ increases by $2$. So, the locality of $W_O(\boldsymbol{\theta})$ is $2d$.      
    \end{proof}

Moreover, from Figure~\ref{fig:w_theta}, we can see that for any $ \boldsymbol{\theta}$, $\widetilde{W}_{O_i}(\boldsymbol{\theta})$ can be computed with cost exponential only in $d$ using tensor contractions of the observable with all the $S$ gates marked in red.

Let $U(\boldsymbol{\theta})$ be the alternating layered ansatz defined  as in Eq.~\eqref{eq:ala}, and $O = \sum_{i = 1} ^ {M} O_i$ be the observable, where each $O_i$ is $1$-local. 
Assume that we are using an iterative optimization algorithm, one that takes as input a target function and outputs its optimizer, to find the maximizer of Eq.~\eqref{eq:vqa-function}, and the whole optimization procedure requires $C$ function evaluations of the form $f_{\rho, O}(\boldsymbol{\theta}^{(1)}),f_{\rho, O}(\boldsymbol{\theta}^{(2)}), \dots, f_{\rho, O}(\boldsymbol{\theta}^{(C)})$. Lemma~\ref{le:also} says that each function evaluation can be seen as estimating the expectation of $2d$-local observables, because $f_{\rho, O}(\boldsymbol{\theta}) =  \sum_i \text{tr}( O_i U(\boldsymbol{\theta}) \rho U(\boldsymbol{\theta}) ^ {\dag}) = \sum_i \text{tr}( W_{O_i}(\boldsymbol{\theta}) \rho )$. 

Now using Theorem~\ref{th:shadows}, we can estimate all the $C$ function evaluations and the whole ALSO algorithm goes as follows:

\begin{enumerate}
    \item Load $T=\mathcal{O}(\log (C) \cdot \text{poly}(n))$ classical shadows of $\rho$. Let $\hat{\rho}_T$ be the shadow state (cf.Theorem~\ref{th:shadows}). 

    \item For all $i$, compute $ \hat{\rho}_T \big|_{A_i}$, where $A_i$ is the sub-register that $ W_{O_i}$ acts non-trivially on.
    \item Use the iterative optimization algorithm to optimize the target function $\hat{f}_{\rho,O}(\boldsymbol{\theta})= \sum_{i}\text{tr}(\widetilde{W}_{O_i}(\boldsymbol{\theta}) \hat{\rho_T} \big|_{A_i} )$.
    
\end{enumerate}

Note that the cost of classical computation is dominated by the computation of $ \sum_{i}\text{tr}(\widetilde{W}_{O_i}(\boldsymbol{\theta}) \hat{\rho_T} \big|_{A_i} )$ and so it scales exponentially only in $d$. Hence, when $d=\mathcal{O}(\log n)$, the classical computational cost scales polynomially on $n$.

    \subsection{Sample complexity}
    In this section, we discuss the sampling complexity of the protocol, that is, the range of values of $T$ that guarantee good estimations of all the function evaluations. We show that when $d = \mathcal{O}(\log n)$, the sample complexity is $\mathcal{O}(\log (C) \cdot \text{poly}(n))$.
    
    \begin{theorem} \label{th:also}
     Let $d, S $ and $U$ be defined as in Eq.~\eqref{eq:ala}. Suppose $\rho$ is an arbitrary $n$-qubit state and $O = \sum_{i = 1} ^ {M} O_i$, where each $O_i$ is an $n$-qubit $1$-local observable.
     Then, for any $\delta, \epsilon \in (0, 1)$ and any $C$ parameter tensors $\boldsymbol{\theta} ^ {(1)}, \boldsymbol{\theta} ^ {(2)}, \ldots, \boldsymbol{\theta} ^ {(C)}$, all values $f_{\rho, O}(\boldsymbol{\theta} ^ {(c)})$ can be estimated using $\hat{f}_{\rho, O}(\boldsymbol{\theta} ^ {(c)}):=\text{tr}( W_O(\boldsymbol{\theta} ^ {(c)}) \hat{\rho}_T)$ 
     with the guarantee 
          \begin{align} \label{eq:approximation_criterion}
            \text{Prob}\left( \bigcap \limits_{c = 1} ^ C \left[ \left| f_{\rho, O}\big(\boldsymbol{\theta} ^ {(c)}\big)\! -\! \hat{f}_{\rho, O}\big(\boldsymbol{\theta} ^ {(c)}\big) \right|  \leq \epsilon \right] \right) \geq 1 - \delta\!
            \end{align}
            where $$ T \geq M ^ 2 \log \left( {\frac{2MC}{\delta}} \right)\cdot \frac{4 ^ {2d + 1}} {\epsilon ^ 2} \max_{i} \| O_i\|_{\infty}^2.$$
    \end{theorem}             
\begin{proof}
          Note that for any $\boldsymbol{\theta} \in \mathbb{R}^{\frac{n}{2} \times d \times p}$, we have $f_{\rho,O}(\boldsymbol{\theta}) = \text{tr}(O U(\boldsymbol{\theta}) \rho U(\boldsymbol{\theta})^{\dag}) = \text{tr}(W_O(\boldsymbol{\theta}) \rho)$. 

        First consider the case when $M = 1$. From Lemma~$\ref{le:also}$, we know that $W_O(\boldsymbol{\theta})$ is a $2d$-local operator, with $\|W_O(\boldsymbol{\theta})\|_{\infty} = \| O\|_{\infty}$. This means that all the function evaluations can be seen as computing expectations of observables with locality $2d$ and $\| \cdot \|_{\infty}$ the same as $O$.

        The function evaluations that we have to approximate are $ f_{\rho, O}(\boldsymbol{\theta} ^ {(1)}), f_{\rho, O}(\boldsymbol{\theta} ^ {(2)}), \dots, f_{\rho, O}(\boldsymbol{\theta} ^ {(C)})$ and hence the observables whose expectation that we have to find are $ W_O(\boldsymbol{\theta}^{(1)}), W_O(\boldsymbol{\theta}^{(2)}), \dots, W_O(\boldsymbol{\theta}^{(C)})$. 
        Then from Theorem~\ref{th:shadows}, by using 
        \begin{align}
            T \geq \log \left( {\frac{2C}{\delta}} \right)\cdot \frac{4 ^ {2d + 1}} {\epsilon ^ 2} \| O\|_{\infty}^2
        \end{align}
        classical shadows of $\rho$, we can estimate the expectation of all of these observables to precision $\epsilon$ with a probability of at least $1 - \delta$. 
        
        When $M > 1$, we no longer have $O$ being necessarily a $1$-local observable. So, for every $c$ we have to estimate the expectation of $\rho$ with the observables $ W_{O_i}(\mathbf{\theta} ^ {(c)}) = U(\boldsymbol{\theta} ^ {(c)}) ^ {\dag} O_i U(\boldsymbol{\theta} ^ {(c)})$ for all $i$, and compute their sum. Hence, the total number of observables that we have to compute expectations with is now $MC$. 
        
        By using
        \begin{align}
            T \geq \log \left( {\frac{2MC}{\delta}} \right)\cdot \frac{4 ^ {2d + 1} } {(\epsilon/M) ^ 2} \max_{i} \| O_i\|_{\infty}^2
        \end{align}
        classical shadows, we will have approximations of $f_{\rho, O_i}(\boldsymbol{\theta} ^ {(c)})$ given as $ \hat{f}_{\rho,O_i}(\boldsymbol{\theta} ^ {(c)}) = \text{tr}(W_{O_i}(\boldsymbol{\theta}) \hat{\rho}_T)$ such that $ \text{Prob}(E)\geq 1 - \delta$ for 
        \begin{align}
            E=\bigcap \limits_{c = 1} ^ C \bigcap \limits_{i = 1} ^ M \left[ \left| f_{\rho, O_i}(\boldsymbol{\theta} ^ {(c)}) - \hat{f}_{\rho, O_i}(\boldsymbol{\theta} ^ {(c)}) \right| \leq \epsilon/M \right].
        \end{align}
        This is because, if we consider the difference $\left| f_{\rho, O_i}(\boldsymbol{\theta} ^ {(c)}) - \hat{f}_{\rho, O_i}(\boldsymbol{\theta} ^ {(c)}) \right| $ being less than or equal to $\epsilon/M$ to be an event, then Theorem~\ref{th:shadows} says that the intersection of these events for all $i,c$ occurs with probability at least $1-\delta$.
        
        Then, with probability at least $1-\delta$, for all $c$ we have 

        \[
            \sum \limits_{i=1}^M \left| f_{\rho, O_i}(\boldsymbol{\theta} ^ {(c)}) - \hat{f}_{\rho, O_i}(\boldsymbol{\theta} ^ {(c)}) \right| \leq \epsilon
        \]
        and thus
        \begin{align*}
            &\ \left| f_{\rho, O}(\boldsymbol{\theta} ^ {(c)}) - \hat{f}_{\rho, O}(\boldsymbol{\theta} ^ {(c)})\right|\\
            = &\ \left| \sum \limits_{i=1}^M \left[ f_{\rho, O_i}(\boldsymbol{\theta} ^ {(c)}) - \hat{f}_{\rho, O_i}(\boldsymbol{\theta} ^ {(c)})\right] \right|\\
            \leq &\ \epsilon 
        \end{align*}
        where $\hat{f}_{\rho, O}(\boldsymbol{\theta} ^ {(c)}) = \sum_{i = 1} ^ M \hat{f}_{\rho, O_i}(\boldsymbol{\theta} ^ {(c)}) $. So all approximations $\hat{f}_{\rho, O}(\boldsymbol{\theta} ^ {(1)}), \hat{f}_{\rho, O}(\boldsymbol{\theta} ^ {(2)}), \dots, \hat{f}_{\rho, O}(\boldsymbol{\theta} ^ {(C)})$ satisfy Eq~\eqref{eq:approximation_criterion}. Furthermore, since each single-copy classical shadow requires only $1$ copy of the input state $\rho$, we can estimate all the function evaluations to the required quality with $T$ copies of $\rho$. 
\end{proof}

    This is remarkable as, without using classical shadows, we may need to estimate $f_{\rho, O}(\boldsymbol{\theta})$ for any parameter tensor $\boldsymbol{\theta}$ through measurements. Suppose $K$ copies of $\rho$ are consumed to estimate each of these values. Then we end up consuming $ CK$ copies of $\rho$, which can be exponentially larger than the number consumed by ALSO. 
    
    In our method, the measurements that we have to make are solely for computing the classical shadows and hence are independent of all $ \boldsymbol{\theta}^{(c)}$. Moreover, each measurement outcome can be reused multiple times thus resulting in a classical memory for efficiently storing the quantum states, for use in alternating layered VQAs. In the standard method of training VQAs, we are not able to reuse the measurement outcomes that are made as part of the optimization, because each measurement outcome is dependent on the input parameter $\boldsymbol{\theta}^{(c)}$. This is a crucial reason as to why ALSO is a much more appealing option to optimize these functions, especially from a practical perspective where one has to do hyperparameter tuning, find the right classical optimizer, etc.

    One important point to note is that even though the constants look large, in practice, we need not necessarily require this many copies (classical shadows) of $\rho$. This will be shortly illustrated in~\ref{sec:simulation_results},  where we are able to match the results of ideal VQA simulations (simulations that use infinite copies of the input state $\rho$) by using a number of copies of $\rho$ orders of magnitude fewer than the number suggested by Theorem~\ref{th:also}.

    \begin{figure}[htb] 

        \centering
        \begin{tabular}{c}
        \includegraphics[width=\columnwidth]{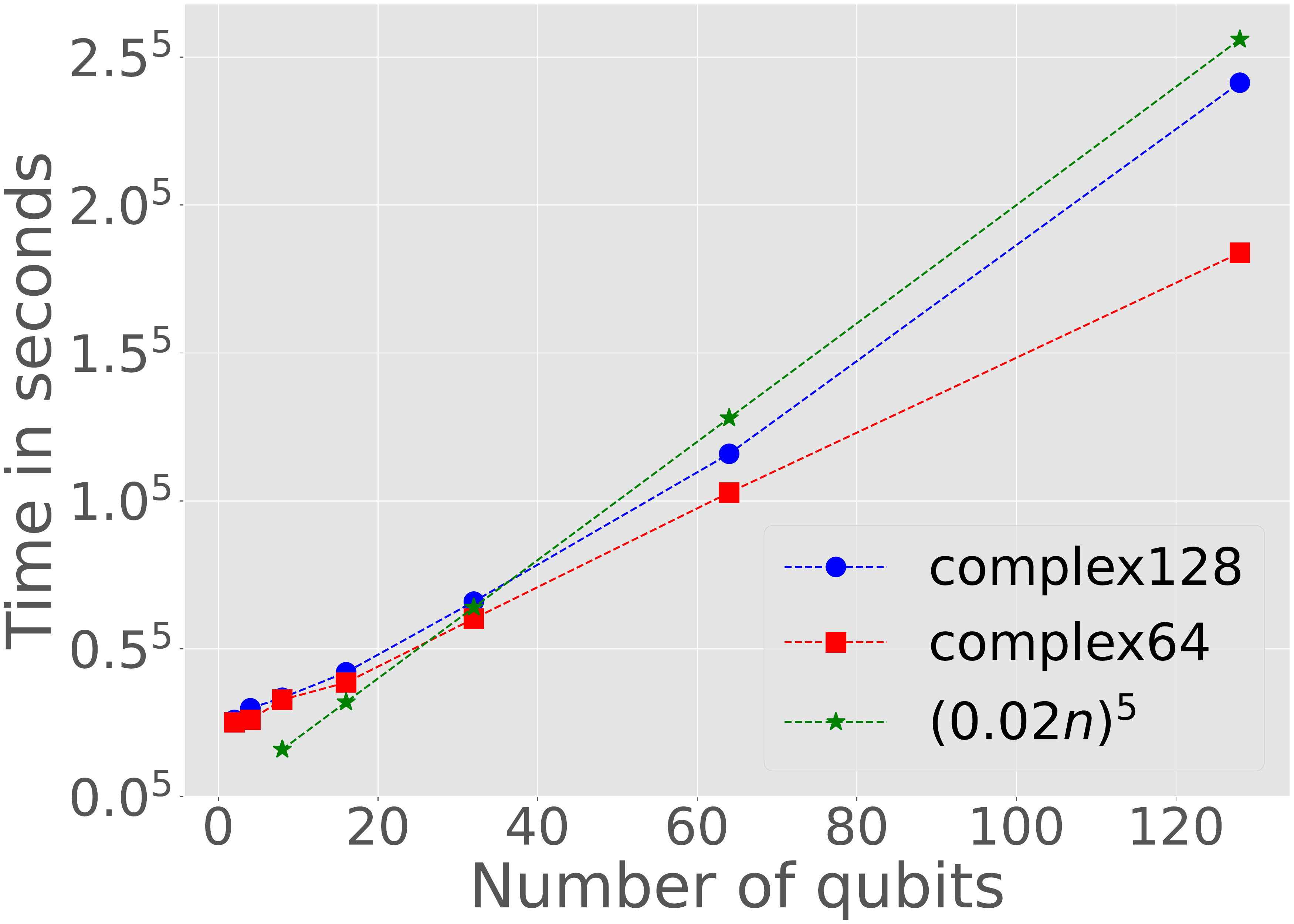}
        \end{tabular}
        \caption{
        Plot showing the time (in seconds) taken for a single function evaluation using ALSO. Here, $d = \lfloor \log n \rfloor$, $S$ is a $2$-qubit parameterized sub-circuit and $O$ is a $1$-qubit observable. Along with the execution times, we plot the function $(0.02n)^5$ to highlight the polynomial dependence of time on the number of qubits. 
        } 
        \label{fig:time_for_one_eval}
    \end{figure}

    The space complexity of the protocol is dominated by the storage of the matrices $\hat{\rho}_T \big|_{A_i}$. Since the dimension of each of these matrices is $4^d$, we need at most $16^dM$ complex numbers to store all of them. For example, let $n, M=50$, $d=5$. Then we see that if we are using $128$ bits to store each complex number, then we only require $838$MB to store all matrices $ \hat{\rho}_{T} \big|_{A_i}$. Time taken for single function evaluations is plotted in Figure~\ref{fig:time_for_one_eval}. On the $x$-axis, we have the number of qubits, and on the $y$-axis, we have the time (in seconds) taken to compute a single function evaluation, averaged over $5$ cases. In each case, $d = \lfloor \log n \rfloor$ and the observable $O$ is a $1$-local observable, with $S$ being the circuit in Figure~\ref{subfig:s}. We plot the results for both `complex$128$' and `complex$64$' being used as datatypes in Python.
    As expected, the plots reflects a polynomial dependence for computational time on the number of qubits.The simulation was carried out on a laptop with $16$GB RAM and $2.6$GHz Intel i$7$ processor. 
    
    One can easily generalize Theorem~\ref{th:also} for arbitrarily local observables and circuits. In a similar setting, if we use $k_0$-local parameterized circuits and an observable that is a sum of $k_1$-local observables, then we can carry out an iterative optimization algorithm with all function evaluations satisfying Eq.~\eqref{eq:approximation_criterion} using \[T \geq \frac{M ^ 2}{\epsilon ^ 2} \cdot \log \left( {\frac{2MC}{\delta}} \right)\cdot 4 ^ {k_1 + (2k_0-2)d - 1}\cdot  \max_{i} \| O_i\|_{\infty}^2 \] copies of the input state. This is because for each increment in depth, the locality of $W_{O_i}(\theta)$ increases by at most $2k_0-2$, starting from $k_1$.

\section{Applications}
\setlabel{Applications}{sec:applications} 

    Here we discuss two practical applications of ALSO in the field of quantum information. For each of these applications, the advantage of using ALSO over the standard VQA will be demonstrated in the next section.

    \subsection{The state preparation problem}
\setlabel{state preparation problem}{sec:spp} 

        Let $\rho = \ket{\psi} \bra{\psi} \in \mathcal{L}(\mathbb{C} ^ {2 ^ n})$ be an $n$-qubit pure quantum state. The state preparation problem of $\rho$ intends to find the parameters of an alternating layered ansatz which best approximates (heuristically) the state $\ket{\psi}$.
        That is, we would like to find a parameter tensor such that $1-\left|\bra{\psi} U(\boldsymbol{\theta}) ^ {\dag} \ket{\z} \right| ^ 2$ is minimized over all $\boldsymbol{\theta}$, where $\ket{\z} \in \mathbb{C}^{2^n}$. This quantity is called the \textit{infidelity} between the states $U(\theta) ^ {\dag}\ket{\z}$ and $\ket{\psi}$. Infidelity measures how ``different'' two states are and admits values in the range $[0,1]$, with an infidelity of $0$ implying that $U(\boldsymbol{\theta}) ^ {\dag}\ket{\z} = \ket{\psi}$. 

        Hence, the problem can be framed as an objective function as $\min_{\boldsymbol{\theta}} \big(1-f_{\ket{\psi} \bra{\psi}, \ket{\z} \bra{\z}}(\boldsymbol{\theta}) \big)$ (or equivalently, $ 1-\max_{\boldsymbol{\theta}} f_{\ket{\psi} \bra{\psi}, \ket{0} \bra{0}}(\boldsymbol{\theta})$), 
        with $\ket{\z} \bra{\z}$  being the observable and $\boldsymbol{\theta} \in \mathbb{R} ^ {\frac{n}{2} \times d \times p}$. However, since $\ket{\z}\bra{\z}$ acts on all $n$-qubits, it does not fit into our framework. Motivated by~\cite{Cerezo2021}, we take $\min_{\boldsymbol{\theta}} (1-f_{\ket{\psi} \bra{\psi}, J}(\boldsymbol{\theta}))$ as the objective function, where $J = \frac{1}{n}\sum_{i = 1} ^ n  \ket{0}_i \bra{0} $ and $\ket{0}_i \bra{0}$ is a $1$-local operator that applies $\ket{0} \bra{0} \in \mathcal{L}(\mathbb{C}^2)$ on the $i^{\text{th}}$ qubit only and $\mathds1$ on all other qubits. In this case, the observable $J$ is a sum of $1$-local observables $O_i=\frac{1}{n}\ket{0}_i \bra{0} $. Hence, this problem fits in our framework and we can use ALSO to optimize it.

    \begin{figure*}[tbh] 
    \centering
    \begin{tabular}{ccc}
         \includegraphics[width=0.65\columnwidth]{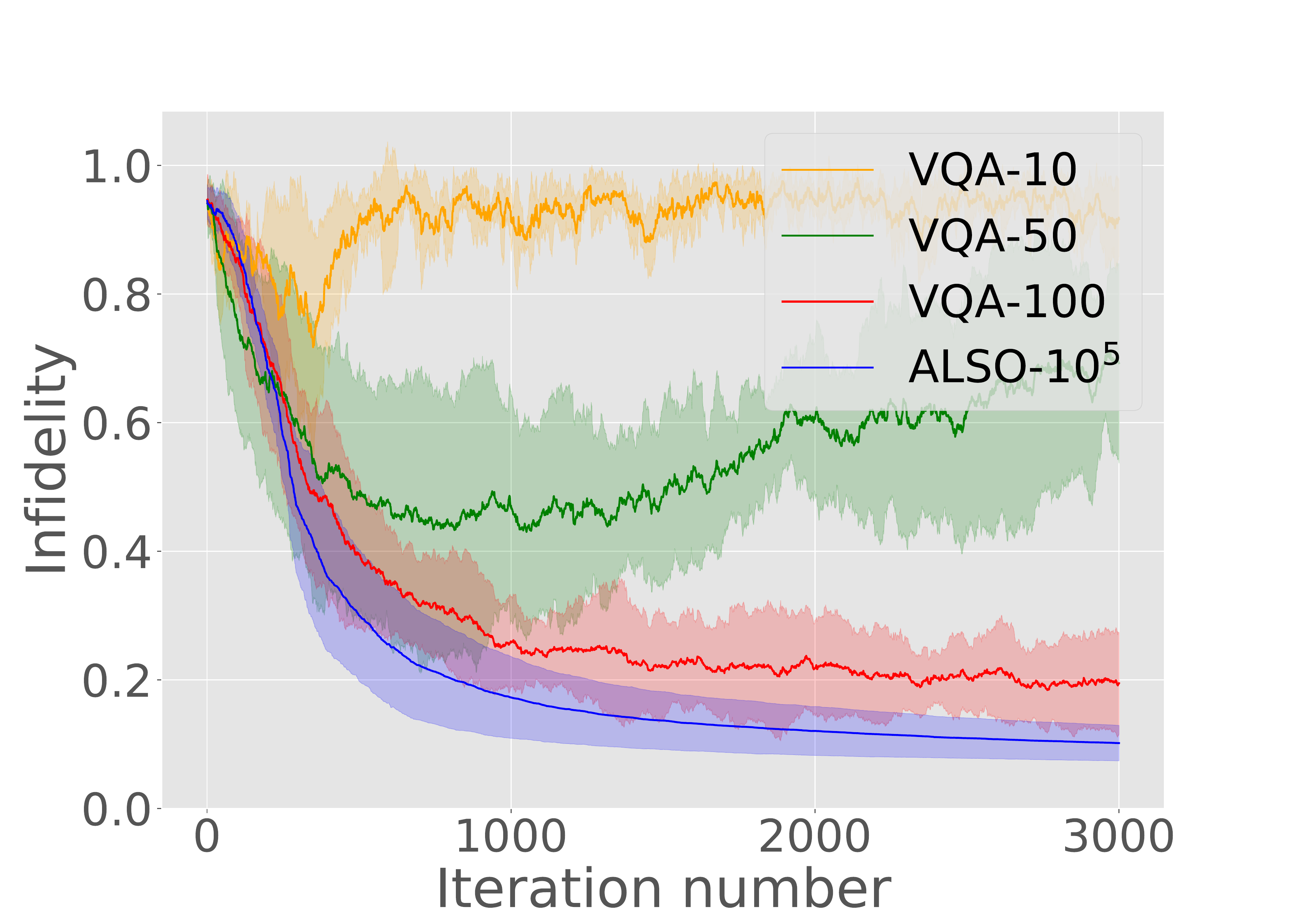} 
         &
        \includegraphics[width=0.65\columnwidth]{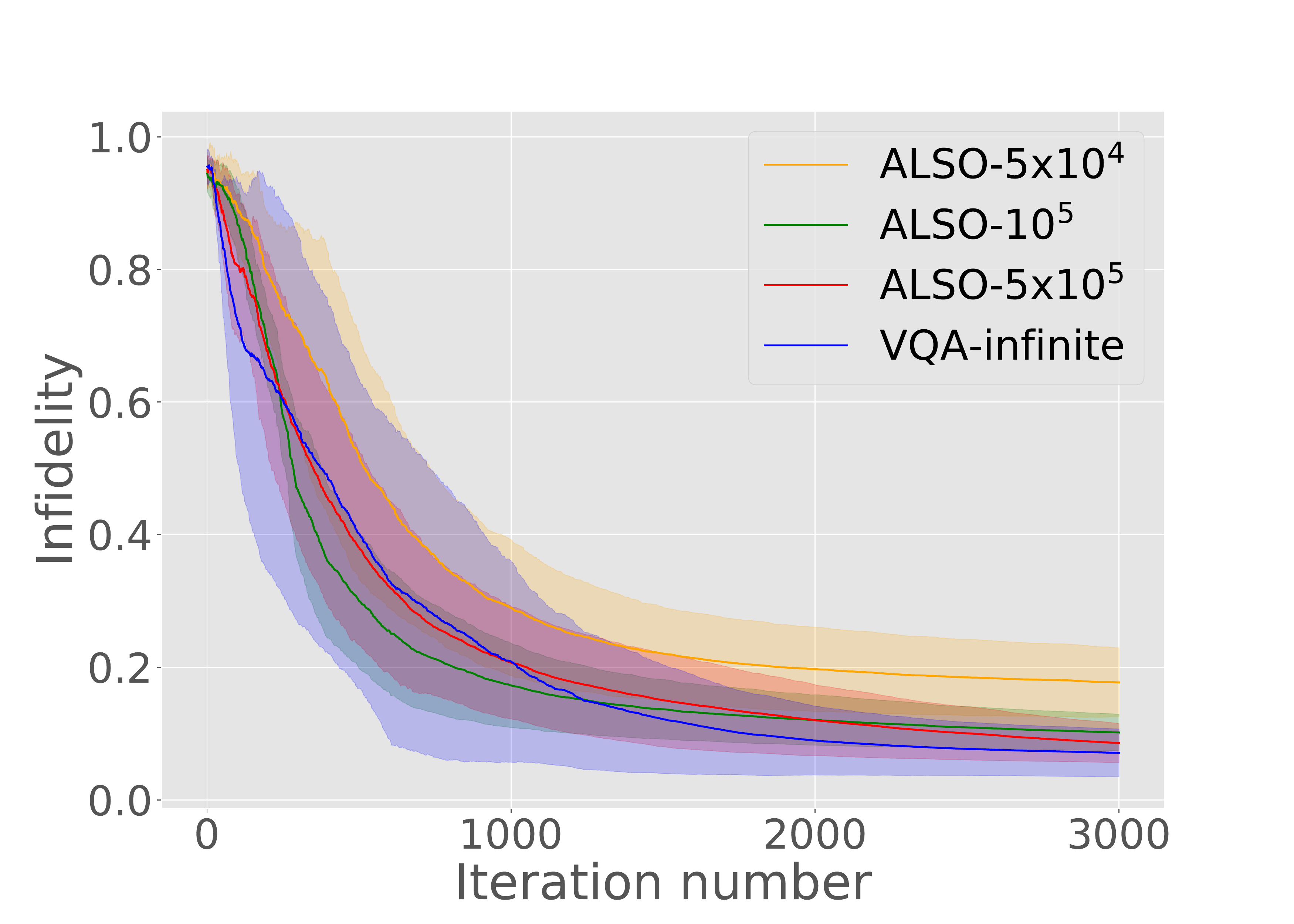}
        &
        \includegraphics[width=0.65\columnwidth]{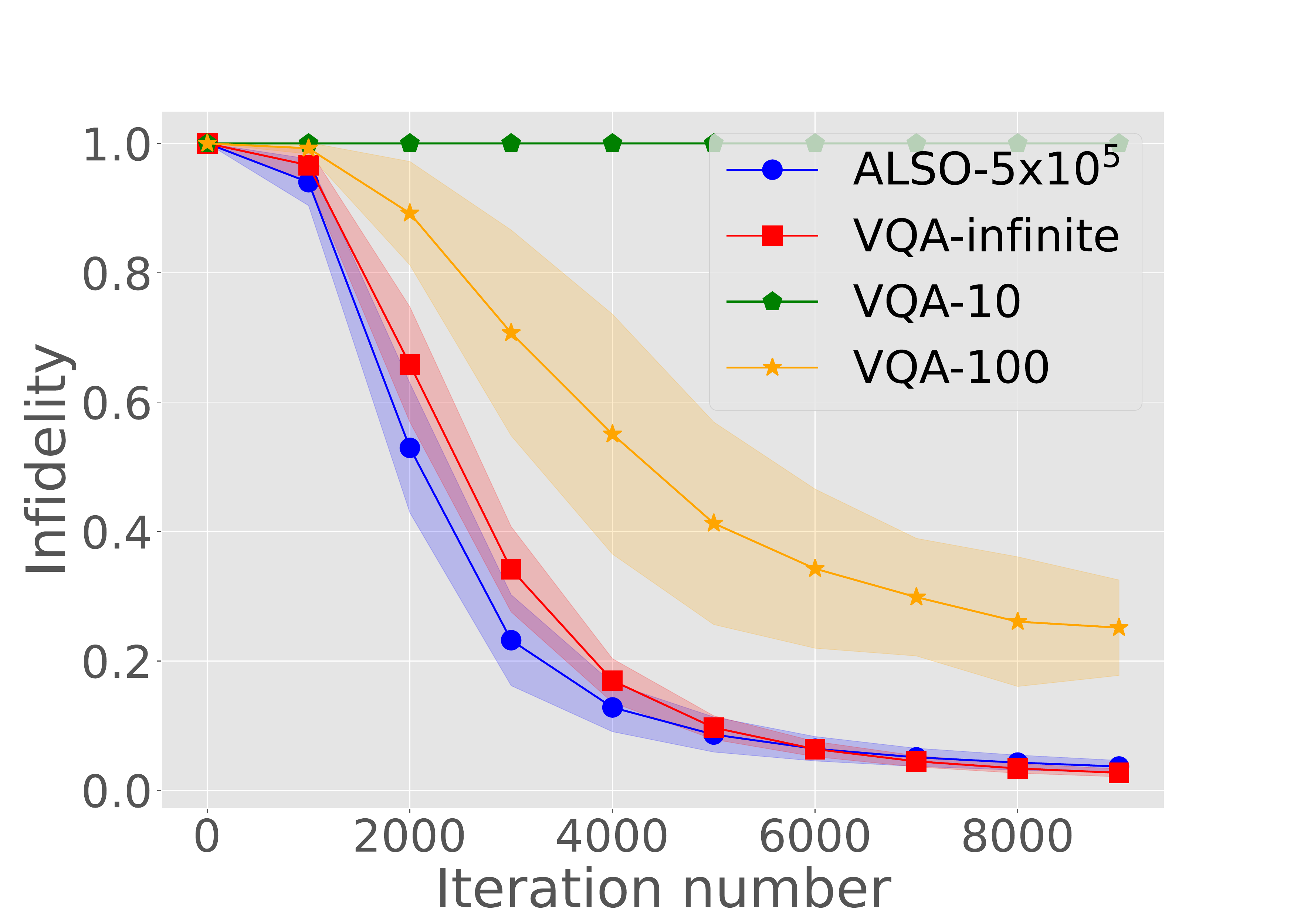} \\
        (a) & (b) & (c) \\ 
         \includegraphics[width=0.65\columnwidth]{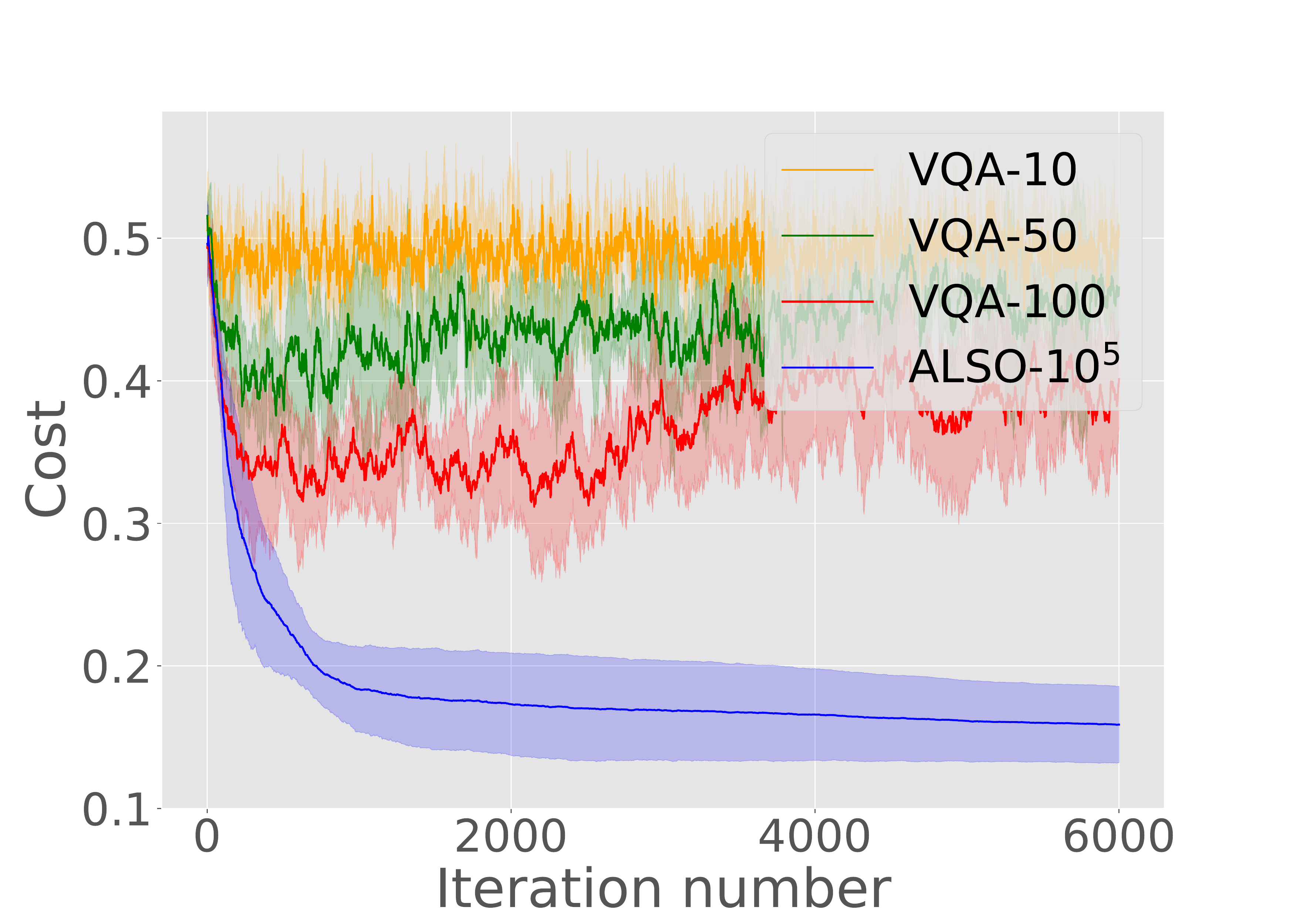} 
         &
         \includegraphics[width=0.65\columnwidth]{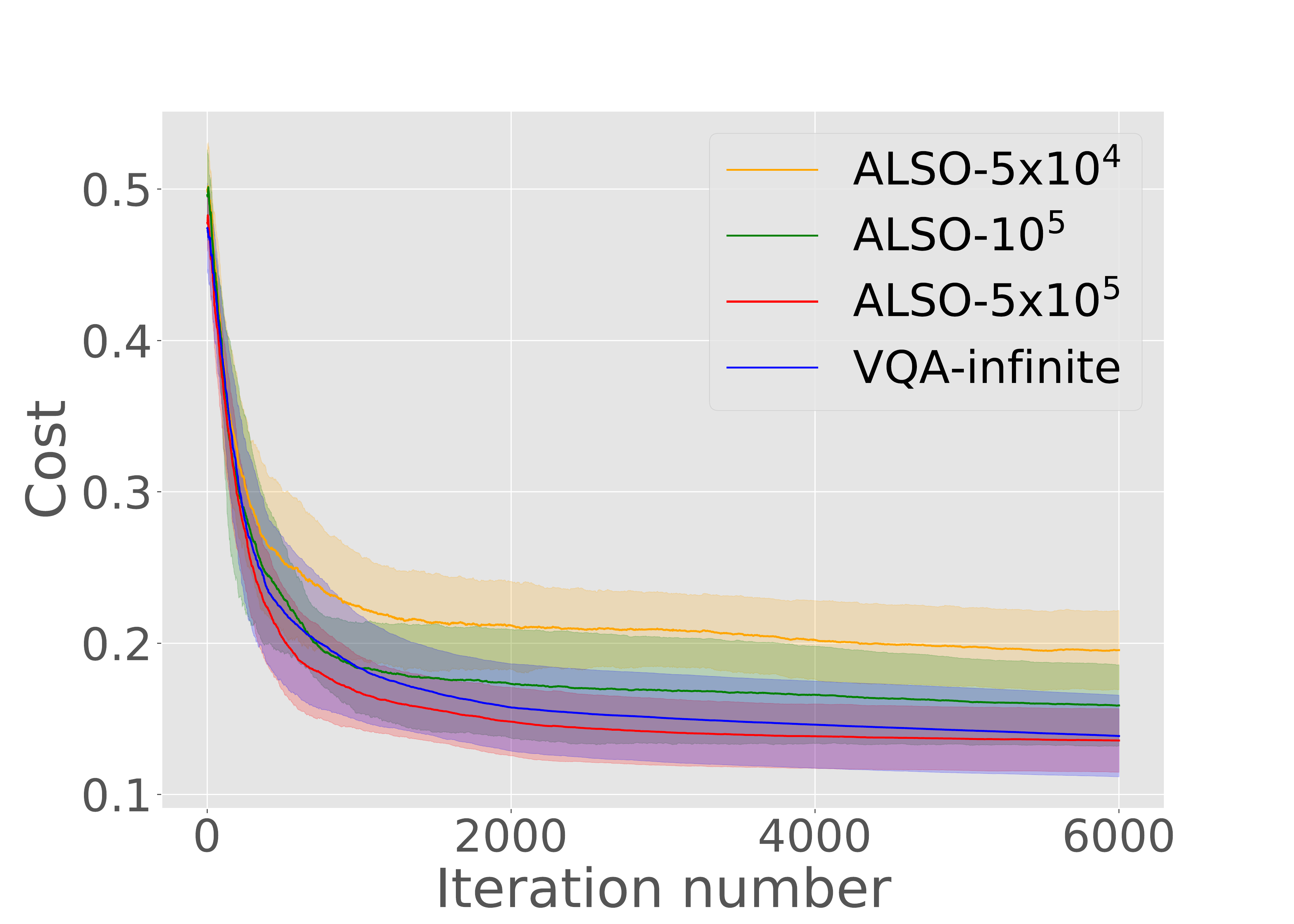}
         &
         \includegraphics[width=0.65\columnwidth]{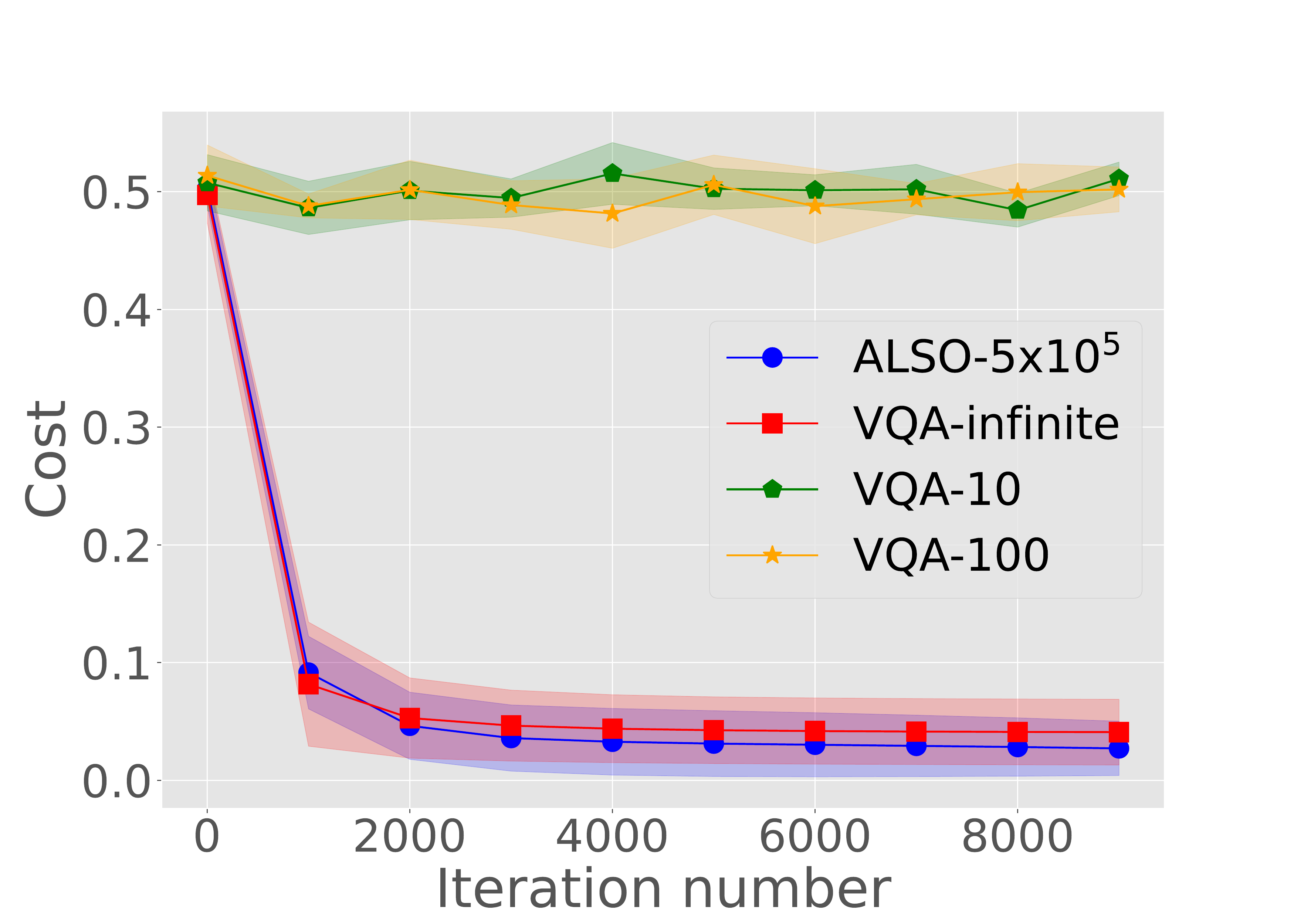} \\
        (d) & (e) & (f) \\ 
    \end{tabular}
    \caption{Simulation results for state preparation (a-c) and quantum autoencoder (d-f) using SPSA. Each graph corresponds to 5 instances of a problem. VQA-$K$ consumes $K$ copies (samples) per function evaluation while ALSO-$T$ consumes $T$ copies (samples) in total. In (a) and (c), we compare the performance of ALSO with standard VQA in the case of $8$-qubit problems. Here, VQA-$10$, VQA-$50$ and VQA-$100$ will consume $4.8 \times 10^5$ ($4.8 \times 10^5$), $2.4 \times 10^6$ ($2.4 \times 10^6$) and $4.8 \times 10^6$ ($4.8 \times 10^6$) copies (samples) respectively while ALSO-$10^5$ consumes only $10^5$ ($10^5$) copies (samples), and still outperforms VQA considerably. Continuing in the $8$-qubit scenario, In (b) and (d), we compare the performance of ALSO with the ideal VQA that consumes infinite copies, and we see that ALSO is able to almost match the results of this ideal VQA using a modest $5 \times 10^5$ ($5 \times 10^5$) copies (samples). In (c) and (f), we plot results of similar experiments carried out on $30$-qubit states. In this case, VQA consumes $5.4 \times 10^6$ ($1.8 \times 10^6 $) and $5.4 \times 10^7$ ($1.8 \times 10^7 $) copies (samples) respectively. Note that here, only iteration numbers that are multiples of $1000$ are plotted.} \label{fig:spp_ae}
    \end{figure*}

    \subsection{Quantum autoencoder} \label{subsec:qae}
        Autoencoder is a popular dimensionality reduction technique in classical machine learning \cite{Hinton1993}. Using deep neural networks, autoencoders learn low dimensional representations of high dimensional input data, which should ideally keep hold of the original characteristics of the data. This can also be seen as a form of data compression.
        Recently, there have been numerous works on extending this concept to quantum data~\cite{Romero2017,Wan2017,Verdon2018,Pepper2019,Lamata2018}. We focus on the version presented in~\cite{Romero2017}, more specifically, its implementation using alternating layered ansatzes described in~\cite{Cerezo2021}.
        
        The idea behind this version of quantum autoencoder is to compress $n$-qubit quantum states into $n_B <n$ qubit states. Consider an ensemble of $n$-qubit states $\mathcal{Z} = \{ (p_i, \ket{\psi_i}) \}$ with each state being prepared in registers $A$ and $B$ having $n_A$ and $n_B$ qubits respectively, where $n=n_A + n_B$. Let $\rho_{\mathcal{Z}} = \sum_{i} p_i \ket{\psi_i} \bra{\psi_i}$. As a measure of the compression effect, we consider $1 - f_{\rho_{\mathcal{Z}}, \mathds1[A] J[B]} (\boldsymbol{\theta})$, where $J= \frac{1}{n_B}\sum_{i = 1} ^ {n_B}  \ket{0}_i \bra{0}$ is defined as in the state preparation problem, to be a sensible cost function that not only forces the population of all the states to be in the qubits in the register $B$, but also involves the same $1$-local observables that we have used for state preparation. Again, this cost function has been
         used in~\cite{Cerezo2021}.

\section{Simulation Results}
\setlabel{Simulation Results}{sec:simulation_results}

In this section, we discuss the experimental results comparing the performance of ALSO and the standard VQA in the two use cases discussed above.

\subsection{Experiments set-up}        
For all experiments, each brick-like sub-circuit $S(\boldsymbol{\theta}_{ij})$ (cf.~Fig.~\ref{fig:ala}) has the form given in Fig.~\ref{subfig:s}. The simulation results presented in this section (except for Table~\ref{table:powell}) have used Simultaneous Perturbation Stochastic Approximation~\cite{Spall1992} (SPSA), where the converging sequences used for state preparation and quantum autoencoder are, respectively, $c_r=a_r=r^{-0.5}$ and $c_r=a_r=r^{-0.3}$.

In the following, we denote by ALSO-$T$ the ALSO algorithm that uses $T$ shadows and by VQA-$K$ the VQA algorithm that consumes per function evaluation $K$ state copies (for state preparation) or $K$ samples from $\mathcal{Z}= \{ (p_i, \ket{\psi_i})\}$  (for quantum autoencoder). In addition, we write VQA-infinite for the VQA algorithm which has access to an infinite number of state copies.

Let $R$ be the total number of iterations of SPSA. Since SPSA requires $2$ function evaluations per iteration, for state preparation and quantum autoencoder, VQA-$K$ will consume $2KRn$ and $2KRn_B$ state copies respectively.
    
\subsection{State preparation experiments}

For the state preparation problem, we first consider the case when $n=8$, i.e., the target state is an 8-qubit state. In each experiment, the target state is compatible with an alternating layered ansatz. We repeat the experiments for five different target states and our results are shown in Fig.~\ref{fig:spp_ae}(a,b), where each plot corresponds to five different instances of a problem. At any value on the $x$-axis, we plot the mean of infidelity/cost values across the five different experiments that were carried out. The coloured area of the plot is marked on top and bottom by the mean plus and minus the standard deviation of the $5$ values at each point respectively.

In Fig.~\ref{fig:spp_ae}(a), VQA-$10$ consumes $2KRn=2 \times 10 \times 3000 \times 8 = 4.8\times 10^5$ state copies and in a similar manner, the other VQA algorithms consume $2.4 \times 10^6$ and $4.8 \times 10^6$ state copies, which are 4.8x, 24x, and 48x of that ALSO consumes. Furthermore, from Fig.~\ref{fig:spp_ae}(b), we can see that ALSO closely matches the outcome of  VQA-infinite with only $5\times 10^5$ state copies.

Moving on from the $8$-qubit scenario, we then carry out similar experiments for 30-qubit systems. Fig.~\ref{fig:spp_ae}(c) shows the results, where, as in \cite{Cerezo2021}, all states involved are computational basis states. 

We note in this case, VQA-$10$ consumes $2 \times 10 \times 9000 \times 30=5.4 \times 10^6$ state copies while ALSO-$T$ remains unchanged with the change of $n$ from 8 to 30. From the figure, it is clear that the similar conclusion for 8-qubit state preparation also applies to 30-qubit state preparation. In particular, ALSO (with $5\times 10^5$ samples) significantly outperforms VQA with 100x more samples. 
        
\subsection{Quantum autoencoder experiments} 
For quantum autoencoder, similar experiments are carried out for both 8- and 30-qubit systems. Ensembles containing two pure states $\ket{\psi_1}$ and $\ket{\psi_2}$ are chosen with $p_1=1/3$ and $p_2=2/3$, and $n_B$ is set as 4 and 10, respectively, for 8- and 30-qubit systems. We repeat the experiments for five different ensembles. The results are summarised in Fig.~\ref{fig:spp_ae}(d-f), which have the same explanation as those in Fig.~\ref{fig:spp_ae}(a-c). Note that the cost values plotted here are the actual cost $f_{\rho,O}(\theta)$ and not their estimations.
From the figure, we can see similar conclusion we have obtained for state preparation also holds for quantum autoencoder. It seems that in this case ALSO with $10^5$ samples significantly outperforms VQA with $48\times 10^5$ samples and ALSO with $5\times 10^5$ samples can often beat VQA-infinite!

\subsection{Resource consumption for the same objective}

\begin{figure}[htb] 
    \centering
    \begin{tabular}{c}
         \includegraphics[width=\columnwidth]{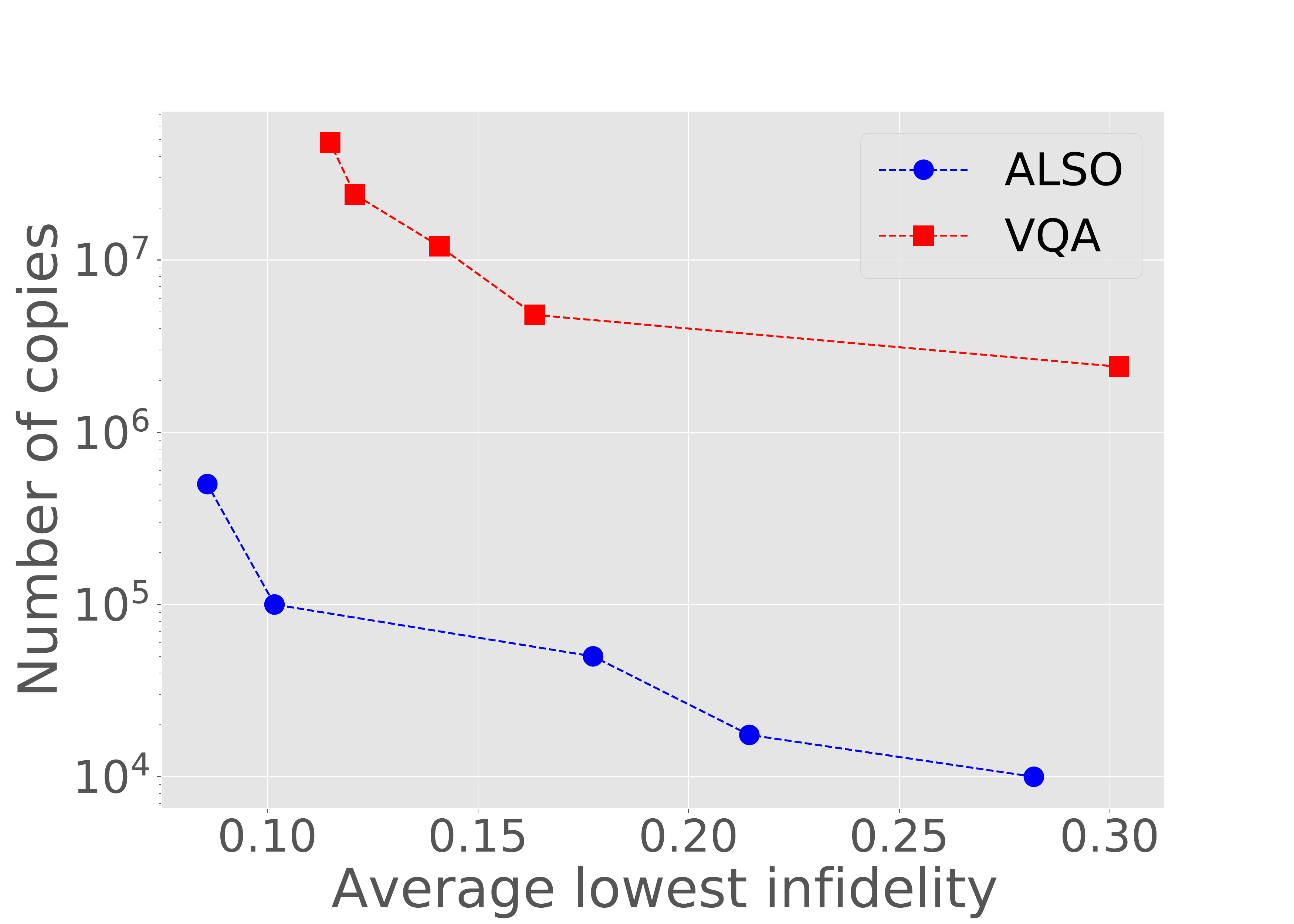} \\ (a) 8-qubit state preparation \\
    \includegraphics[width=\columnwidth]{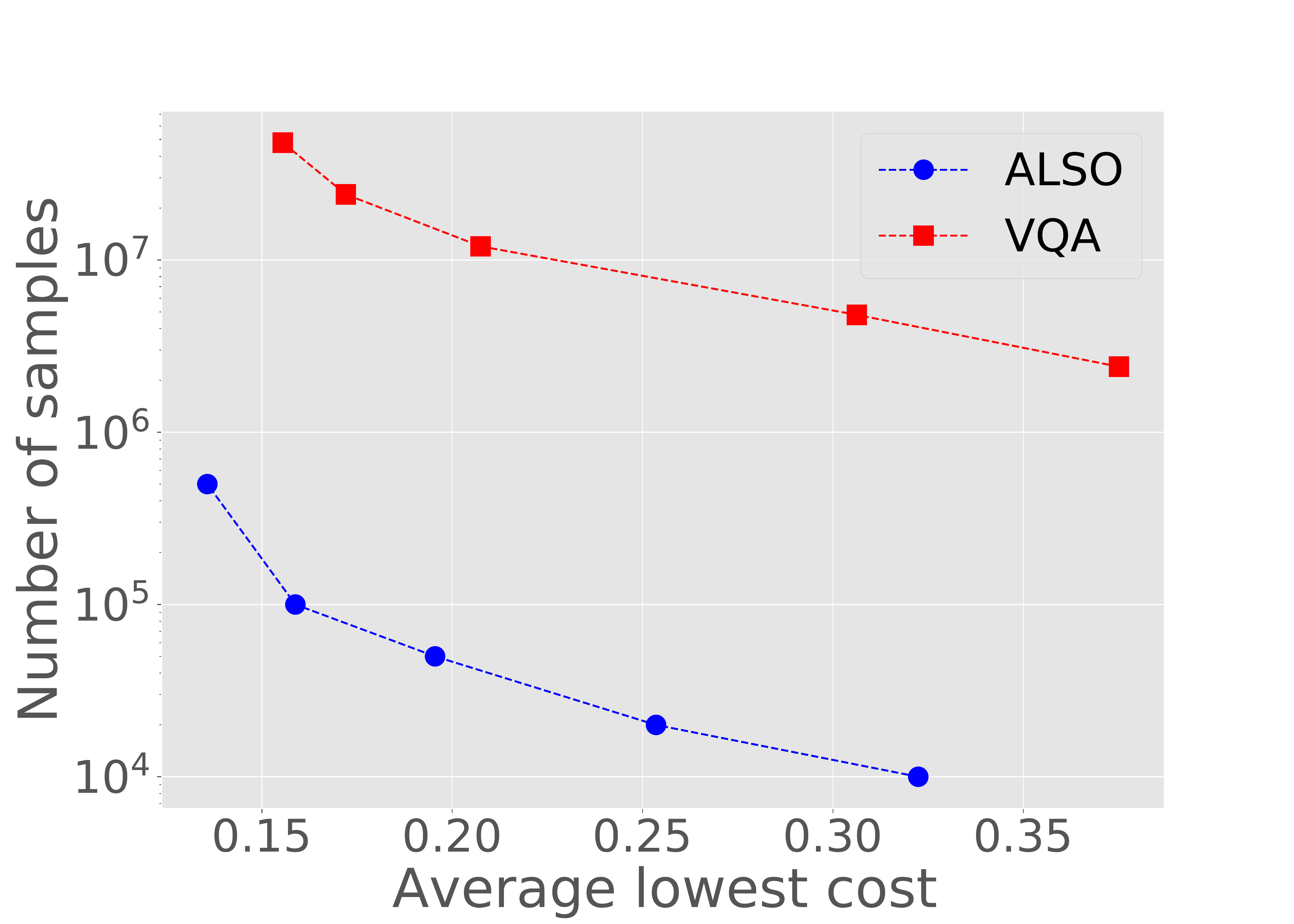} \\
    (b) 8-qubit quantum autoencoder
    \end{tabular}
    \caption{Resource requirement for different objectives. On the $x$-axis, we plot the least average infidelity (cost) of $5$ instances of the corresponding problem. On the $y$-axis, we plot the number of copies (samples) that were required to achieve them using SPSA. We see that ALSO achieves order of magnitude savings in the number of copies (samples).} \label{fig:fid_vs_meas_8q}

\end{figure}

In above, we compared the performance of ALSO and VQA algorithms with predetermined  resources. The efficiency of ALSO over VQA can also be illustrated by comparing the resource consumption required for the same objective. Given an objective, which can be either the average lowest infidelity or the average lowest cost, we carry out experiments to check how many state copies or samples are required for ALSO or VQA to achieve the objective. The results are presented in Fig.~\ref{fig:fid_vs_meas_8q}, where each point represents the average of five instances.  It is clear that ALSO achieves a huge advantage in the number of state copies or samples that were required to achieve the specific levels of quality. 

\subsection{More iterations by using Powell's method}
        
All the simulation results discussed above have used SPSA to find the optimal parameters. In each case, we set the maximum iterations to be the same for ALSO and VQA. From a practical point of view, this is unfair as ALSO can carry out more iterations with the given number of state copies or samples.

To further demonstrate this advantage of ALSO, we turn to Powell's method~\cite{Powell1964} to optimize the parameters. We carry out $8$-qubit state preparation as well as quantum autoencoder optimizations and the results are presented in Table~\ref{table:powell}. In the infidelity (cost) columns, each entry is an average optimal infidelity (cost) of $5$ instances of the problem, and we give an approximation of the average number of copies consumed (except for ALSO where exactly $5\times 10^5$ copies are consumed) to achieve these values in the \#copies (\#samples) columns.

We set $5\times10^4$ as an upper limit on the total number of function evaluations for VQA. But, since ALSO does not consume any copies for more iterations, we don't set any limit in the case of ALSO. As we can see, our approach greatly outperforms VQA in this case. Interestingly, in the case of VQA, the optimizers terminated in $5\times10^3-3\times10^4$ function evaluations in most cases. Only for the state preparation problem and with $10^5$ copies consumed per function evaluation, we saw the optimizer exceeding the $5\times10^4$ limit. We also observe that VQA with Powell's method performs very poorly compared to VQA with SPSA when $K$ is small, which is possibly due to the inherent ability of SPSA to deal with noisy functions.

\begin{table}
\scalebox{0.9}{
  \begin{tabular}{c|c|c|c|c}
 &
      \multicolumn{2}{c|}{state preparation  } &
      \multicolumn{2}{c}{quantum autoencoder}  \\  \hline
    & \#copies & infidelity & \#samples & cost  \\ \hline
    VQA-$10^2$ & $5\times10^5$ & 0.921  & $5.6\times10^5$& 0.494 \\ 
             VQA-$10^3$ &$1.3\times10^7$ & 0.348 &$4.6\times 10^6$ & 0.408 \\ 
             VQA-$10^4$ &$3.3\times 10^8$ & 0.094 &$10^8$ & 0.250 \\ 
             VQA-$10^5$ & $4 \times 10^9$& 0.069 & $2.1\times10^9$ & 0.188 \\
             ALSO & $5\times10^5$ & 0.004 & $5\times10^5$ & 0.117
  \end{tabular}
  }
        \caption{Simulation results comparing the performance of ALSO and the standard VQA when using Powell's method, where infidelity (cost) is the average lowest infidelity (cost) of five instances, and \#copies (\#samples) is the number of state copies (samples) used by the algorithm.}
        \label{table:powell}
\end{table}

\section{Conclusion and Future direction}
\setlabel{Conclusion and Future direction}{sec:conclusion}

In this work, we proposed ALSO --- an efficient method to train alternating layered VQAs that is exponentially better than the standard way of training VQAs in terms of the number of copies of input state consumed (or, in some other applications, number of executions of the quantum computer). The saving of state copies is especially useful when multiple rounds of the same optimization algorithm are required for various choices of hyperparameters, or when one has to experiment with different algorithms altogether. Moreover, ALSO is implementable using fewer and simpler quantum operations; in fact, only single-qubit measurements according to Pauli bases (in the classical shadow preparation stage) are required in ALSO. Another interesting benefit of the classical shadow technique, and so ALSO in particular, is that the produced classical shadows can be reused in different (independent) tasks. For example, the same set of classical shadows can be used in both finding the state preparation circuits and building quantum autoencoders.
	
In terms of future directions, we are trying to design similar resource efficient protocols for other trainable ansatzes such as the Quantum CNN ansatz used in~\cite{Cong2019,Pesah2021}. Another topic of interest would be to extend this method to techniques in other areas of quantum information that use large amounts of copies or executions such as device calibration and error correction.

\section{Code availability}
    Source code to implement ALSO and replicate the results presented in this work can be found \href{https://github.com/afradnyf/ALSO}{here}.
\section{Acknowledgement}
    We thank Afham and Richard Kueng for valuable discussions. 
    This work was partially supported by the Australian Research Council (Grant Nos: DP180100691 and DP220102059). AB was partially supported by the Sydney Quantum Academy PhD scholarship.  		

\bibliography{references}

\begin{thebibliography}{32}
\providecommand{\natexlab}[1]{#1}

\bibitem[{Aaronson(2018)}]{Aaronson2018}
Aaronson, S. 2018.
\newblock Shadow Tomography of Quantum States.
\newblock In \emph{Proceedings of the 50th Annual ACM SIGACT Symposium on
  Theory of Computing}, STOC 2018, 325–338. New York, NY, USA: Association
  for Computing Machinery.
\newblock ISBN 9781450355599.

\bibitem[{Arrasmith et~al.(2021)Arrasmith, Cerezo, Czarnik, Cincio, and
  Coles}]{Arrasmith2021}
Arrasmith, A.; Cerezo, M.; Czarnik, P.; Cincio, L.; and Coles, P.~J. 2021.
\newblock Effect of barren plateaus on gradient-free optimization.
\newblock \emph{{Quantum}}, 5: 558.

\bibitem[{Cerezo et~al.(2021{\natexlab{a}})Cerezo, Arrasmith, Babbush,
  Benjamin, Endo, Fujii, McClean, Mitarai, Yuan, Cincio, and
  Coles}]{Cerezo2021_VQA}
Cerezo, M.; Arrasmith, A.; Babbush, R.; Benjamin, S.~C.; Endo, S.; Fujii, K.;
  McClean, J.~R.; Mitarai, K.; Yuan, X.; Cincio, L.; and Coles, P.~J.
  2021{\natexlab{a}}.
\newblock Variational quantum algorithms.
\newblock \emph{Nature Reviews Physics}, 3(9): 625--644.

\bibitem[{Cerezo et~al.(2021{\natexlab{b}})Cerezo, Sone, Volkoff, Cincio, and
  Coles}]{Cerezo2021}
Cerezo, M.; Sone, A.; Volkoff, T.; Cincio, L.; and Coles, P.~J.
  2021{\natexlab{b}}.
\newblock Cost function dependent barren plateaus in shallow parametrized
  quantum circuits.
\newblock \emph{Nature Communications}, 12(1).

\bibitem[{Chow, Dial, and Gambetta(2021)}]{Chow2021}
Chow, J.; Dial, O.; and Gambetta, J. 2021.
\newblock IBM Quantum breaks the 100-qubit processor barrier.
\newblock \emph{IBM Research Blog}.

\bibitem[{Cong, Choi, and Lukin(2019)}]{Cong2019}
Cong, I.; Choi, S.; and Lukin, M.~D. 2019.
\newblock Quantum convolutional neural networks.
\newblock \emph{Nature Physics}, 15(12): 1273--1278.

\bibitem[{Gily{\'{e}}n, Song, and Tang(2022)}]{Gilyen2022}
Gily{\'{e}}n, A.; Song, Z.; and Tang, E. 2022.
\newblock An improved quantum-inspired algorithm for linear regression.
\newblock \emph{Quantum}, 6: 754.

\bibitem[{Gilyén, Lloyd, and Tang(2018)}]{Gilyen2018}
Gilyén, A.; Lloyd, S.; and Tang, E. 2018.
\newblock Quantum-inspired low-rank stochastic regression with logarithmic
  dependence on the dimension.
\newblock \emph{arXiv:1811.04909}.

\bibitem[{Hadfield et~al.(2019)Hadfield, Wang, O{\textquotesingle}Gorman,
  Rieffel, Venturelli, and Biswas}]{Hadfield2019}
Hadfield, S.; Wang, Z.; O{\textquotesingle}Gorman, B.; Rieffel, E.; Venturelli,
  D.; and Biswas, R. 2019.
\newblock From the Quantum Approximate Optimization Algorithm to a Quantum
  Alternating Operator Ansatz.
\newblock \emph{Algorithms}, 12(2): 34.

\bibitem[{Hinsche et~al.(2021)Hinsche, Ioannou, Nietner, Haferkamp, Quek,
  Hangleiter, Seifert, Eisert, and Sweke}]{Hinsche2021}
Hinsche, M.; Ioannou, M.; Nietner, A.; Haferkamp, J.; Quek, Y.; Hangleiter, D.;
  Seifert, J.-P.; Eisert, J.; and Sweke, R. 2021.
\newblock Learnability of the output distributions of local quantum circuits.
\newblock \emph{arXiv:2110.05517}.

\bibitem[{Hinton and Zemel(1993)}]{Hinton1993}
Hinton, G.~E.; and Zemel, R.~S. 1993.
\newblock Autoencoders, Minimum Description Length and Helmholtz Free Energy.
\newblock In \emph{Proceedings of the 6th International Conference on Neural
  Information Processing Systems}, NIPS'93, 3–10. San Francisco, CA, USA:
  Morgan Kaufmann Publishers Inc.

\bibitem[{Huang et~al.(2021{\natexlab{a}})Huang, Broughton, Mohseni, Babbush,
  Boixo, Neven, and McClean}]{Huang2021_power}
Huang, H.-Y.; Broughton, M.; Mohseni, M.; Babbush, R.; Boixo, S.; Neven, H.;
  and McClean, J.~R. 2021{\natexlab{a}}.
\newblock Power of data in quantum machine learning.
\newblock \emph{Nature Communications}, 12(1).

\bibitem[{Huang, Kueng, and Preskill(2020)}]{Huang2020}
Huang, H.-Y.; Kueng, R.; and Preskill, J. 2020.
\newblock Predicting many properties of a quantum system from very few
  measurements.
\newblock \emph{Nature Physics}, 16(10): 1050--1057.

\bibitem[{Huang et~al.(2021{\natexlab{b}})Huang, Kueng, Torlai, Albert, and
  Preskill}]{Huang2021}
Huang, H.-Y.; Kueng, R.; Torlai, G.; Albert, V.~V.; and Preskill, J.
  2021{\natexlab{b}}.
\newblock Provably efficient machine learning for quantum many-body problems.
\newblock \emph{arXiv:2106.12627}.

\bibitem[{Kandala et~al.(2017)Kandala, Mezzacapo, Temme, Takita, Brink, Chow,
  and Gambetta}]{Kandala2017}
Kandala, A.; Mezzacapo, A.; Temme, K.; Takita, M.; Brink, M.; Chow, J.~M.; and
  Gambetta, J.~M. 2017.
\newblock Hardware-efficient variational quantum eigensolver for small
  molecules and quantum magnets.
\newblock \emph{Nature}, 549: 242--246.

\bibitem[{Lamata et~al.(2018)Lamata, Alvarez-Rodriguez, Mart{\'{\i}}n-Guerrero,
  Sanz, and Solano}]{Lamata2018}
Lamata, L.; Alvarez-Rodriguez, U.; Mart{\'{\i}}n-Guerrero, J.~D.; Sanz, M.; and
  Solano, E. 2018.
\newblock Quantum autoencoders via quantum adders with genetic algorithms.
\newblock \emph{Quantum Science and Technology}, 4(1): 014007.

\bibitem[{Li, Song, and Wang(2021)}]{Li2021}
Li, G.; Song, Z.; and Wang, X. 2021.
\newblock VSQL: Variational Shadow Quantum Learning for Classification.
\newblock \emph{Proceedings of the AAAI Conference on Artificial Intelligence},
  35(9): 8357--8365.

\bibitem[{McClean et~al.(2018)McClean, Boixo, Smelyanskiy, Babbush, and
  Neven}]{Mcclean2018}
McClean, J.~R.; Boixo, S.; Smelyanskiy, V.~N.; Babbush, R.; and Neven, H. 2018.
\newblock Barren plateaus in quantum neural network training landscapes.
\newblock \emph{Nature Communications}, 9(1).

\bibitem[{Mitarai et~al.(2018)Mitarai, Negoro, Kitagawa, and
  Fujii}]{Mitarai2018}
Mitarai, K.; Negoro, M.; Kitagawa, M.; and Fujii, K. 2018.
\newblock Quantum circuit learning.
\newblock \emph{Phys. Rev. A}, 98: 032309.

\bibitem[{Nakaji and Yamamoto(2021)}]{Nakaji2021}
Nakaji, K.; and Yamamoto, N. 2021.
\newblock Expressibility of the alternating layered ansatz for quantum
  computation.
\newblock \emph{Quantum}, 5: 434.

\bibitem[{Pepper, Tischler, and Pryde(2019)}]{Pepper2019}
Pepper, A.; Tischler, N.; and Pryde, G.~J. 2019.
\newblock Experimental Realization of a Quantum Autoencoder: The Compression of
  Qutrits via Machine Learning.
\newblock \emph{Phys. Rev. Lett.}, 122: 060501.

\bibitem[{Pesah et~al.(2021)Pesah, Cerezo, Wang, Volkoff, Sornborger, and
  Coles}]{Pesah2021}
Pesah, A.; Cerezo, M.; Wang, S.; Volkoff, T.; Sornborger, A.~T.; and Coles,
  P.~J. 2021.
\newblock Absence of Barren Plateaus in Quantum Convolutional Neural Networks.
\newblock \emph{Physical Review X}, 11(4).

\bibitem[{Powell(1964)}]{Powell1964}
Powell, M. J.~D. 1964.
\newblock {An efficient method for finding the minimum of a function of several
  variables without calculating derivatives}.
\newblock \emph{The Computer Journal}, 7(2): 155--162.

\bibitem[{Romero, Olson, and Aspuru-Guzik(2017)}]{Romero2017}
Romero, J.; Olson, J.~P.; and Aspuru-Guzik, A. 2017.
\newblock Quantum autoencoders for efficient compression of quantum data.
\newblock \emph{Quantum Science and Technology}, 2(4): 045001.

\bibitem[{Sack et~al.(2022)Sack, Medina, Michailidis, Kueng, and
  Serbyn}]{Sack22}
Sack, S.~H.; Medina, R.~A.; Michailidis, A.~A.; Kueng, R.; and Serbyn, M. 2022.
\newblock Avoiding Barren Plateaus Using Classical Shadows.
\newblock \emph{{PRX} Quantum}, 3(2).

\bibitem[{Slattery, Villalonga, and Clark(2022)}]{Slattery2022}
Slattery, L.; Villalonga, B.; and Clark, B.~K. 2022.
\newblock Unitary block optimization for variational quantum algorithms.
\newblock \emph{Phys. Rev. Research}, 4: 023072.

\bibitem[{Spall(1992)}]{Spall1992}
Spall, J. 1992.
\newblock Multivariate stochastic approximation using a simultaneous
  perturbation gradient approximation.
\newblock \emph{IEEE Transactions on Automatic Control}, 37(3): 332--341.

\bibitem[{Tang(2019)}]{Tang2019}
Tang, E. 2019.
\newblock A Quantum-Inspired Classical Algorithm for Recommendation Systems.
\newblock In \emph{Proceedings of the 51st Annual ACM SIGACT Symposium on
  Theory of Computing}, STOC 2019, 217–228. New York, NY, USA: Association
  for Computing Machinery.
\newblock ISBN 9781450367059.

\bibitem[{Tang(2021)}]{Tang2021}
Tang, E. 2021.
\newblock Quantum Principal Component Analysis Only Achieves an Exponential
  Speedup Because of Its State Preparation Assumptions.
\newblock \emph{Physical Review Letters}, 127(6).

\bibitem[{Verdon, Pye, and Broughton(2018)}]{Verdon2018}
Verdon, G.; Pye, J.; and Broughton, M. 2018.
\newblock A Universal Training Algorithm for Quantum Deep Learning.
\newblock \emph{arXiv:1806.09729}.

\bibitem[{Wan et~al.(2017)Wan, Dahlsten, Kristj{\'{a}}nsson, Gardner, and
  Kim}]{Wan2017}
Wan, K.~H.; Dahlsten, O.; Kristj{\'{a}}nsson, H.; Gardner, R.; and Kim, M.~S.
  2017.
\newblock Quantum generalisation of feedforward neural networks.
\newblock \emph{npj Quantum Information}, 3(1).

\bibitem[{Wu et~al.(2021)Wu, Li, Wang, Feng, Ding, and Xie}]{Wu2021}
Wu, A.; Li, G.; Wang, Y.; Feng, B.; Ding, Y.; and Xie, Y. 2021.
\newblock Towards Efficient Ansatz Architecture for Variational Quantum
  Algorithms.
\newblock \emph{arXiv:2111.13730}.

\end{thebibliography}

\end{document}